%% file: main.tex
\documentclass[journal]{IEEEtran}
\usepackage{cite}
\usepackage{amsmath,amssymb,amsfonts}
\usepackage{graphicx}
\usepackage{textcomp}
\usepackage{xcolor}
\usepackage{multirow}
\usepackage{makecell}
\usepackage{algorithm,algcompatible}
\usepackage{hhline}
\usepackage{multicol}
\usepackage{amssymb}
\usepackage{amsmath}
\usepackage{amsthm}
\usepackage{siunitx}
\usepackage{calc}
\usepackage{array}
\usepackage{optidef}
\usepackage{epstopdf}
\usepackage{amsmath}
\usepackage{mathtools}
\usepackage{mathrsfs}
\usepackage{bm}
\usepackage{cite}
\usepackage{boldline}
\usepackage{algorithm}
\usepackage{booktabs}
\usepackage{siunitx}
\usepackage{subcaption}

\input{macros}

\def\BibTeX{{\rm B\kern-.05em{\sc i\kern-.025em b}\kern-.08em
    T\kern-.1667em\lower.7ex\hbox{E}\kern-.125emX}}

\newtheorem{theorem}{Theorem}

\newtheorem{corollary}{Corollary}
    
\begin{document}
\title{Scalable Rate-Splitting Precoding via Recurrent Structure-Preserving Graph Neural Networks}

\author{Wonseok~Choi,~\IEEEmembership{Student Member,~IEEE}, Jeongjae~Lee,~\IEEEmembership{Student Member,~IEEE}, and Songnam~Hong,~\IEEEmembership{Senior Member,~IEEE}
\thanks{W. Choi, J. Lee, and S. Hong are with the Department of Electronic Engineering, Hanyang University, Seoul 04763, Korea (e-mail: \{ryan4975, lyjcje7466, snhong\}@hanyang.ac.kr).}
\vspace{-0.1cm}
}

\maketitle

\begin{abstract}
Graph neural network~(GNN)-based precoding has demonstrated strong potential for scalable multi-user beamforming in multi-user multiple-input single-output~(MU-MISO) systems under space division multiple access~(SDMA). However, direct extension to rate-splitting multiple access~(RSMA) is non-trivial due to the coupled common/private-stream structure inherent to RSMA, which requires a fundamentally different graph representation and permutation equivariance structure. Motivated by this, we propose a recurrent structure-preserving graph neural network~(RS-GNN) for scalable RSMA precoding. RS-GNN constructs precoder-dependent graph features at every refinement layer, enabling closed-loop interference-aware message passing, and recovers the common and private precoders through an analytically grounded structure-based reconstruction via a differentiable linear solver. This design decouples the learnable parameters from fixed system dimensions, enabling generalization to unseen system sizes without retraining. We formally prove that RS-GNN satisfies mixed permutation equivariance with respect to both user and antenna orderings, and show that RS-GNN reduces to conventional SDMA precoding as a special case by deactivating the common-stream branch. Simulation results demonstrate that RS-GNN achieves near-WMMSE sum-rate performance with significantly lower online inference time, while generalizing robustly to unseen system sizes; its SDMA special case consistently outperforms existing GNN-based precoders across unseen antenna and user configurations, SNR regimes, and channel distributions.
\end{abstract}

\begin{IEEEkeywords} 
Rate-splitting multiple access, MU-MISO precoding, graph neural networks, structure-preserving learning, precoding optimization, size generalization.
\end{IEEEkeywords}

\section{Introduction}
\label{sec:introduction}

Rate-splitting multiple access~(RSMA) has emerged as a powerful downlink transmission strategy that provides flexible interference management beyond conventional space division multiple access~(SDMA) and non-orthogonal multiple access~(NOMA). In RSMA, each user message is split into common and private parts: the common parts of all users are jointly encoded into a single common-stream decoded by all users via successive interference cancellation~(SIC), after which each user decodes its own private-stream. By allowing each receiver to partially decode interference and partially treat it as noise, RSMA provides an additional degree of freedom for managing multi-user interference, bridging and outperforming SDMA and NOMA across a wide range of channel conditions~\cite{clerckx2023primer,mao2022rate,mao2018rate}. Despite these advantages, RSMA precoding leads to a challenging non-convex optimization problem due to the joint design of a common precoder, private precoders, and a common-rate allocation vector that are mutually coupled through a min-rate bottleneck constraint. Existing works address this problem through iterative optimization algorithms which solve a sequence of tractable subproblems to obtain a locally optimal or stationary solution~\cite{mao2018rate,li2020rate,mishra2021rate,mao2021rate,park2022rate}. However, these methods must be re-executed from scratch at every channel realization and require many iterations to converge, resulting in high online computational complexity that limits their deployment in dynamic wireless environments.

This motivates learning-based precoding methods that approximate the optimal channel state information~(CSI)-to-precoder mapping offline and deploy it online with a single forward pass, amortizing the per-channel cost of iterative optimization~\cite{sun2018learning}.
However, the RSMA precoding mapping is highly nonlinear and high-dimensional. The outputs are strongly coupled through multi-user interference and a min-rate bottleneck constraint, requiring the network to learn over a large hypothesis space~\cite{sun2022improving}. 
Purely data-driven architectures therefore suffer from poor sample efficiency and limited generalization to unseen system sizes or channel distributions~\cite{zhao2022learning,sun2022learning}. These limitations highlight the need for structured learning frameworks that go beyond black-box function approximation by embedding the analytical properties of the RSMA precoding problem into the neural network architecture, while retaining scalability to varying system configurations.

\subsection{Related Work and Limitations}

Optimizing the RSMA precoder is more challenging than its SDMA counterpart, as it requires the joint design of a common precoder shared by all users and private precoders subject to a common-rate bottleneck constraint~\cite{clerckx2023primer}. The classical iterative WMMSE algorithm~\cite{mao2018rate} achieves near-optimal performance but incurs high per-channel computational complexity, as each iteration requires solving a convex subproblem via optimization toolboxes such as CVX.
A more efficient alternative is the fractional programming with hyperplane fixed-point iteration~(FP-HFPI) algorithm~\cite{fang2024rate}, which first derives the optimal beamforming structure of the RSMA precoder via the Karush--Kuhn--Tucker (KKT) conditions and then iterates only over the low-dimensional Lagrangian dual variables to recover the full precoder analytically. By avoiding optimization toolboxes and exploiting the optimal beamforming structure, FP-HFPI achieves comparable sum-rate performance to WMMSE with significantly lower computational complexity. Despite this advantage, both WMMSE and FP-HFPI must be executed from scratch at every channel realization, incurring a per-channel online cost that grows with the number of iterations. This computational burden motivates learning-based alternatives.

RS-BNN~\cite{wang2024rs} is a representative learning-based approach that unfolds the FP-HFPI algorithm by replacing its inner optimization loop with a lightweight deep neural network~(DNN) to predict the Lagrangian dual variables, thereby amortizing the per-channel optimization cost.  Unlike SDMA, where WMMSE admits a relatively straightforward unfolding structure, the min-rate bottleneck constraint in RSMA couples the common and private precoders, making direct WMMSE unfolding substantially more complex; FP-HFPI is therefore preferred as the unfolding target due to its toolbox-free structure and analytically derived optimal beamforming solution. However, RS-BNN relies on a fixed-dimensional architecture whose input/output dimensions are tied to the training system configuration, limiting its generalization to unseen numbers of antennas or users. Moreover, the learned updates remain confined to the FP-HFPI alternating optimization trajectory, restricting the solution space to the neighborhood of the solver's stationary point.

To overcome the large hypothesis space and limited size scalability of fixed-dimensional neural precoders, graph neural networks~(GNNs) have emerged as a structured learning architecture for wireless precoding~\cite{zhao2024understanding}. Multi-user interference naturally admits a graph representation in which users correspond to nodes and interference relationships correspond to edges. By sharing message passing functions across nodes and employing permutation-invariant aggregation, GNN-based precoders inherently preserve user permutation equivariance and handle different numbers of users with the same trainable parameters. Building on these properties, several GNN-based precoding methods have been developed for conventional SDMA systems~\cite{shen2022graph, guo2023model,zhang2025gradient, sun2025homogeneous, chowdhury2023deep}.

However, directly extending these SDMA-oriented GNN precoders to RSMA presents two fundamental architectural challenges. Existing SDMA-GNN designs produce only user-indexed private precoder outputs and provide no mechanism to generate the user-invariant common precoder required by RSMA; the coupled common/private output structure demands a dedicated architectural design that standard GNN combination functions do not enforce.
Moreover, most existing GNN precoders construct graph features solely from the channel realization or from algorithmic variables inherited from an unfolded solver, and do not re-evaluate the interference pattern induced by the currently reconstructed common and private precoders at each refinement step. Consequently, simply replacing the SDMA objective with the RSMA sum-rate does not yield a compatible architecture.

Several recent studies have investigated GNN-based learning for RSMA systems~\cite{chen2024energy,kumar2025graph,huang2025gnn, huang2026max,hou2025sum,kumar2025self}, demonstrating the potential of graph-based learning for RSMA resource allocation and precoder optimization. However, existing RSMA-GNN approaches are mainly designed for specific system configurations and rely on neural readout modules to generate optimization variables or precoders from channel-dependent graph embeddings. Consequently, the interference pattern induced by the evolving precoders is never fed back into message passing, and joint antenna- and user-size-generalizable RSMA precoding with a single trained model is not addressed.

\subsection{Main Contributions}

In this paper, we propose a recurrent structure-preserving graph neural network~(RS-GNN) for scalable RSMA precoding. The main contributions are summarized as follows.

\begin{itemize}

\item \textit{Closed-loop interference-aware message passing.}
Unlike existing GNN-based precoders that fix graph features based on the channel realization alone, RS-GNN reconstructs precoder-dependent graph features at every refinement layer from the current precoding state, allowing message passing to explicitly track the desired and interference beam patterns induced by the evolving common and private precoders. This enables closed-loop refinement that is unavailable in standard feed-forward GNN designs.

\item \textit{Structure-preserving precoder recovery.}
Instead of relying on generic node-wise neural updates to produce high-dimensional precoding vectors, RS-GNN predicts only compact user-wise structural variables and recovers the full RSMA precoder through an analytically grounded reconstruction based on the KKT-optimal beamforming structure. This structural constraint substantially narrows the hypothesis space and decouples the learnable neural modules from the antenna dimension, leading to stable end-to-end training via the sum-rate objective without requiring labeled outputs from numerical solvers.

\item \textit{Size-generalizable inference with formal equivariance guarantee.}
The combination of shared interference-aware message passing and structure-based reconstruction allows a single trained RS-GNN model to operate across varying numbers of users and antennas without modification or retraining. We formally prove that RS-GNN satisfies mixed permutation equivariance with respect to both user and antenna orderings, ensuring that the learned precoding policy is invariant to arbitrary system indexing. Extensive simulations validate that this design yields consistent near-WMMSE performance across unseen system sizes with significantly reduced online inference time.

\item \textit{Reduction to SDMA as a special case.} By deactivating the common-stream branch, RS-GNN reduces to an SDMA precoder without any architectural modification. Simulations in SDMA mode confirm that the proposed structure-preserving design consistently outperforms existing GNN-based SDMA precoders across varying system sizes, SNR regimes, and channel distributions.
\end{itemize}

\subsection{Outline}
The remainder of this paper is organized as follows. Section~\ref{sec:pre} describes the MU-MISO RSMA system model and formulates the sum-rate maximization problem. Section~\ref{sec:proposed} presents the proposed RS-GNN architecture and establishes its theoretical properties, including mixed permutation equivariance and reduction to SDMA as a special case. Section~\ref{sec:TrainInfer} describes the training and inference procedures of RS-GNN. Section~\ref{sec:experiments} presents simulation results, including ablation studies, comparison with RSMA benchmark methods, and comparison with GNN-based baselines under the SDMA special case. Section~\ref{sec:conclusion} concludes the paper.

\vspace{0.1cm}
{\em Notations.}
We use $\mathbf a$ and $\mathbf A$ to denote a column vector and a matrix, respectively. For a vector $\mathbf a$, $\operatorname{diag}(\mathbf a)$ denotes a diagonal matrix whose diagonal entries are the elements of $\mathbf a$. The operators $(\cdot)^{\rm T}$, $(\cdot)^{\rm H}$, $|\cdot|$, $\|\cdot\|_F$, $\mathbb E[\cdot]$, $\mathcal O(\cdot)$, $\Re(\cdot)$, and $\Im(\cdot)$ represent transpose, Hermitian transpose, absolute value of a scalar, the Frobenius norm, expectation, big-$\Oc$ complexity order, real part, and imaginary part, respectively. The sets $\mathbb C^{m\times n}$ and
$\mathbb R^{m\times n}$ denote the complex- and real-valued spaces of dimension $m\times n$, respectively. Especially, $\mathbb R_+^{m\times n}$ denotes the non-negative real-valued spaces. The $m$-dimensional all-one vector (resp. all-zero vector) and
$m\times m$ identity matrix are denoted by $\mathbf 1_m$ (resp. $\mathbf 0_m$) and $\mathbf I_m$, respectively. We use $\mathcal{CN}(\boldsymbol\mu,\mathbf\Sigma)$ to denote the circularly symmetric complex Gaussian distribution with mean $\boldsymbol\mu$ and covariance $\mathbf\Sigma$.

\section{Preliminaries}\label{sec:pre}

In this section, we present the system model for downlink multi-user multiple-input single-output~(MU-MISO) precoding with RSMA and formulate the corresponding sum-rate maximization problem. We then review the standard GNN-based precoding framework and identify three fundamental limitations that prevent its direct extension to size-generalizable RSMA precoding.

\subsection{System Model and Problem Formulation}

We consider a downlink MU-MISO system in which a base station (BS) equipped with $N$ transmit antennas serves $K$ single-antenna users indexed by $k\in \mathcal{K}\triangleq \{1,2,\dots,K\}$. The downlink channel vector from the BS to user $k$ is denoted by $\mathbf h_k\in \mathbb C^{N}$, and the composite channel matrix is defined as 
\begin{equation}
    \mathbf{H} = [\mathbf{h}_1,\mathbf{h}_2, \dots, \mathbf{h}_K ] \in \mathbb{C}^{N \times K}.
\end{equation}
The BS employs one-layer RSMA by transmitting a common-stream~$s_{\rm c}\in\CC$ and $K$ private-streams $\sv_{\rm p}=[s_{{\rm p},1},s_{{\rm p},2},\dots,s_{{\rm p},K}]^{\rm T}\in\CC^{K}$. Letting $\sv = \big[s_{\rm c}, \sv_{\rm p}^{\rm T}\big]^{\rm T}\in\CC^{K+1}$, all streams are assumed to be mutually independent with zero mean and unit variance, i.e., $\EE\big[\sv\sv^{\rm H}\big]=\Id_{K+1}$. Let $\mathbf w_{\rm c}\in\mathbb C^{N}$ and $\Wm_{\rm p} = [\mathbf w_{{\rm p},1}, \wv_{{\rm p},2}, \dots,\wv_{{\rm p},K}]\in\mathbb C^{N\times K}$ denote the precoders for the common-stream and private-streams, respectively, and define the aggregate precoding matrix $\mathbf W=[\mathbf w_{\rm c}, \Wm_{\rm p}]\in \mathbb C^{N\times (K+1)}$. The transmit signal is $\xv = \Wm\sv\in\CC^{N}$, and the received signal at user~$k$ is
\begin{equation}
y_k=\mathbf h_k^{\rm H}\mathbf x+z_k,
\end{equation}
where $z_k\sim\mathcal {CN}\big(0,\sigma_k^2\big)$ is additive Gaussian noise with variance $\sigma_k^2>0$.


Each user first decodes the common-stream while treating all private-streams as interference, performs SIC to remove the decoded common-stream from the received signal, and then decodes its own private-stream. The signal-to-interference-plus-noise ratios~(SINRs) at user $k$ for decoding the common and private-streams are defined as
\begin{align}
\mathrm{SINR}_{{\rm c}, k}(\mathbf W) &= \frac{\big|{\mathbf h}_k^{\rm H} {\mathbf w}_{\rm c}\big|^2}{\sum_{i\in\mathcal{K}} \big|{\mathbf h}_k^{\rm H} {\mathbf w}_{{\rm p},i}\big|^2 + \sigma_k^2},\\
\mathrm{SINR}_{{\rm p},k}(\mathbf W) &= \frac{\big|{\mathbf h}_k^{\rm H} {\mathbf w}_{{\rm p},k}\big|^2}{\sum_{i\in\mathcal{K}\setminus \{k\}} \big|{\mathbf h}_k^{\rm H} {\mathbf w}_{{\rm p},i}\big|^2 + \sigma_k^2},
\end{align} respectively. The corresponding achievable rates are
\begin{align}
R_{{\rm c},k}(\mathbf W)&=\log_2\big(1+\mathrm{SINR}_{{\rm c},k}(\mathbf W)\big),\\ R_{{\rm p},k}(\mathbf W)&=\log_2\big(1+\mathrm{SINR}_{{\rm p},k}(\mathbf W)\big).
\end{align}
Since the common-stream must be decodable by every user, its achievable rate is limited by the weakest common-stream decoder \cite{fang2024rate}: $R_{\mathrm c}(\mathbf{W})=\min_{k\in\mathcal{K}}R_{{\mathrm c},k}(\mathbf W)$. The RSMA sum-rate therefore reduces to
\begin{equation}
  R_{\rm sum}(\mathbf{W})
  = R_{\rm c}(\mathbf{W})+\sum_{k\in\mathcal{K}}R_{{\rm p},k}(\mathbf{W}),\label{eq:rsum}
\end{equation}
and the sum-rate maximization problem is formulated as
\begin{subequations}
\begin{align}
  \mathcal{P}:
  \max_{\mathbf{W}}\;&
  R_{\rm sum}(\mathbf{W})\label{eq:mainprob}
  \\
  \text{s.t.}\;&
  \|\mathbf{W}\|_F^2 \leq P_{\rm TX},\label{eq:mainprob1}
\end{align}
\end{subequations} where $P_{\rm TX}$ denotes the total transmit power budget. For a given channel matrix $\mathbf{H}$, we define the optimal precoding policy for the problem $\Pc$ as a mapping $f^{\star}:\mathbb{C}^{N\times K}\to \mathbb{C}^{N\times(K+1)}$ such that
\begin{equation}
  \mathbf{W}^{\star}=f^{\star}(\mathbf{H}),
\end{equation}
where $\mathbf{W}^{\star}\in\CC^{N\times (K+1)}$ denotes the optimal precoders.

\begin{figure}[t]
\centering
\includegraphics[width=1\linewidth]{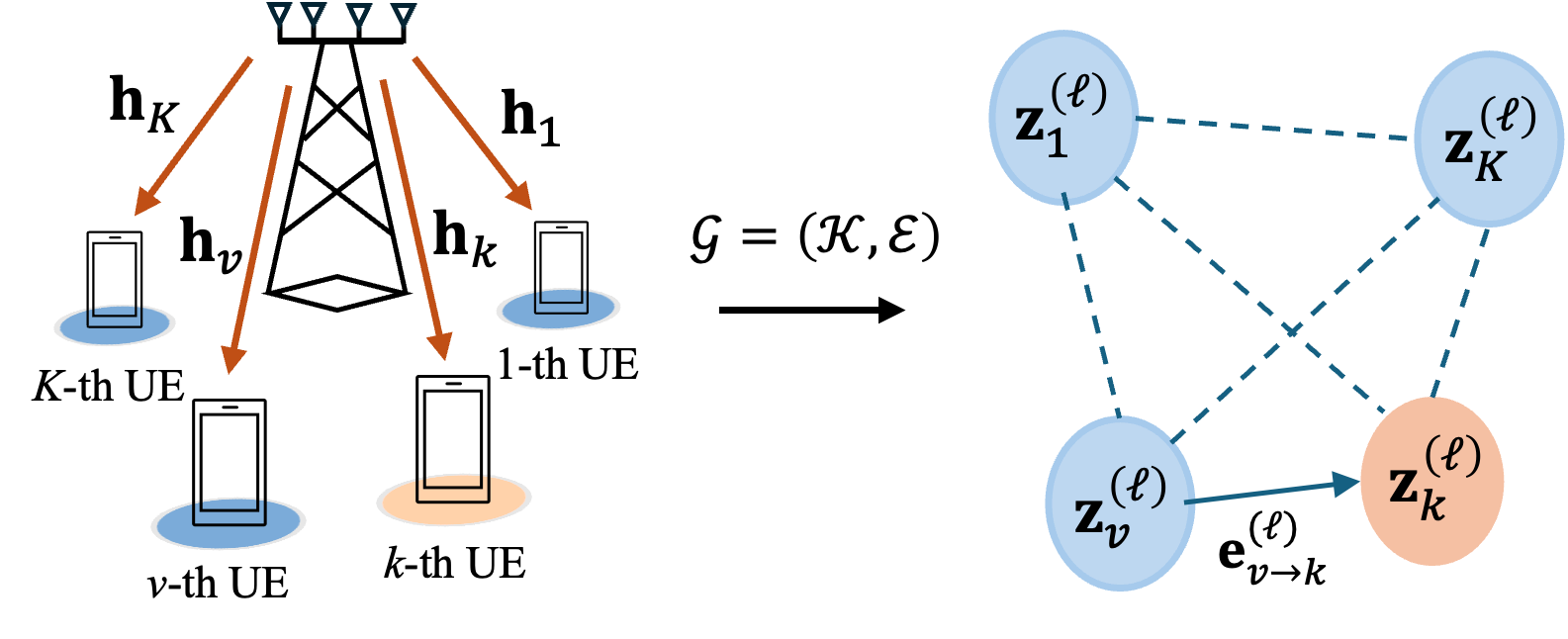}
\caption{Graph representation of the MU-MISO downlink precoding problem, where each node represents a user and each directed edge represents an inter-user interference link.} 
\label{fig:graph}
\end{figure}

\subsection{Standard GNN-Based Precoding and Its Limitations}\label{subsec:standardGNN}
The problem~$\mathcal{P}$ is non-convex and, in general, does not admit a closed-form solution. To reduce the real-time computational burden of solving $\mathcal{P}$, therefore, learning-based methods can approximate the optimal policy $f^{\star}(\cdot)$ with a parameterized function. Among these, GNNs have emerged as a natural architecture for multi-user precoding. By sharing learnable functions across graph elements and aggregating messages through permutation-invariant operations, GNNs inherently preserve user permutation equivariance and can process graphs of varying sizes. Accordingly, the learned GNN-based precoding policy is written as
\begin{equation}\label{eq:gnn_policy}
\hat{\mathbf{W}}=\hat{f}_{\boldsymbol{\theta}}(\mathbf{H}), 
\end{equation} where $\hat{\Wm}\in\CC^{N\times(K+1)}$ denotes the learned precoding matrix, $\hat{f}_{\boldsymbol{\theta}}(\cdot)$ denotes the approximated and parameterized precoding policy, and ${\thetav}$ denotes the learnable parameters.

For GNN-based precoding, the system is modeled as a fully connected directed graph $\mathcal{G}=(\mathcal{K},\mathcal{E})$, as illustrated in Fig.~\ref{fig:graph}, where the node set $\mathcal{K}$ represents the $K$ users and the directed edge set $\mathcal{E}$ captures the inter-user interference relationships. A standard GNN updates the hidden representation of each node through three operations: i) processing, ii) pooling, and iii) combination \cite{shen2022graph}.
{Let $\mathbf{z}_k^{(\ell)}$ denote the hidden state of a node $k$ at the $\ell$-th GNN layer of total $L$ layers, where $\ell\in\{0,1,\dots,L-1\}$}. For a target node $k$, the message from a neighboring node  $v\in\mathcal{K}\setminus\{k\}$ at a layer $\ell$ is generated by a processing function:
\begin{equation}
\mathbf{m}_{v\to k}^{(\ell)} = \phi_{\rm proc}^{(\ell)}\Big(\mathbf{z}_k^{(\ell)},\mathbf{z}_v^{(\ell)},\mathbf{e}_{v\to k}^{(\ell)}\Big),
\end{equation} where $\mathbf{e}_{v\to k}^{(\ell)}$ is the edge feature of the directed edge $(v,k)\in\Ec$. The incoming messages are aggregated by a permutation-invariant pooling function:
\begin{equation}
\mathbf{a}_k^{(\ell)} =  \operatorname{POOL}\Big(\mathbf{m}_{v\to k}^{(\ell)}\Big),
\end{equation}
where $\operatorname{POOL}(\cdot)$ can be implemented by summation, mean, maximum, or attention-weighted summation of all incoming messages to the node $k$, i.e., $\mathbf{m}_{v\to k}^{(\ell)},\forall v\in\mathcal{K}\setminus\{k\}$. The pooling function is permutation-invariant with respect to the ordering of neighboring nodes, ensuring that the aggregated message is independent of the user indexing. Finally, the hidden state of node $k$ is updated by a combination function:
\begin{equation}
\label{eq:combination}
\mathbf{z}_k^{(\ell+1)} = \phi_{\rm comb}^{(\ell)}\Big(\mathbf{z}_k^{(\ell)}, \mathbf{a}_k^{(\ell)}\Big).
\end{equation}
Note that $\phi_{\rm proc}^{(\ell)}(\cdot)$ and $\phi_{\rm comb}^{(\ell)}(\cdot)$ are learnable and shared across all nodes.

While parameter sharing and permutation-invariant pooling ensure user permutation equivariance, we identify three fundamental limitations that prevent standard GNNs from achieving robust size-generalizable RSMA precoding. When the hidden states or output precoders are parameterized with a fixed antenna dimension, the learned processing and combination functions become implicitly tied to the training antenna size, limiting generalization to unseen antenna configurations without retraining. RSMA also requires a mixed output structure absent in conventional SDMA precoding: the private precoders must be permutation equivariant with respect to user ordering, whereas the common precoder must remain invariant to user reordering, a distinction that a standard combination function in~\eqref{eq:combination} does not explicitly enforce. 
Moreover, most existing GNN-based precoders construct graph features solely from the channel realization $\mathbf{H}$.
For example, a representative edge feature can be written as $e_{v\to k}^{(\ell)}=\big|\hv_{k}^{\rm H}\hv_{v}\big|$ which is fixed across all message passing layers~\cite{shen2022graph}. Consequently, the interference pattern induced by the currently refined precoder is never re-evaluated during message passing, resulting in an open-loop design.

This motivates a GNN architecture that simultaneously addresses all three limitations, as detailed in the following sections.

\section{Proposed RS-GNN for RSMA Precoding}
\label{sec:proposed}

In this section, we propose the recurrent structure-preserving GNN architecture, referred to as {\bf RS-GNN}, which is built on two design principles that jointly address the limitations identified in Section~\ref{subsec:standardGNN}. Rather than generating the full precoding matrix as a high-dimensional neural output, RS-GNN predicts only low-dimensional user-wise structural variables through shared message passing and recovers the common and private precoders analytically via structure-preserving linear systems, decoupling the learnable components from fixed system dimensions. The recovered precoders are then fed back into the graph at every refinement layer to reconstruct precoder-dependent graph features, enabling closed-loop beam pattern-aware message passing that is absent in standard feed-forward GNN designs. 

\begin{figure*} [t]
    \centering
    \includegraphics[width=1\linewidth]{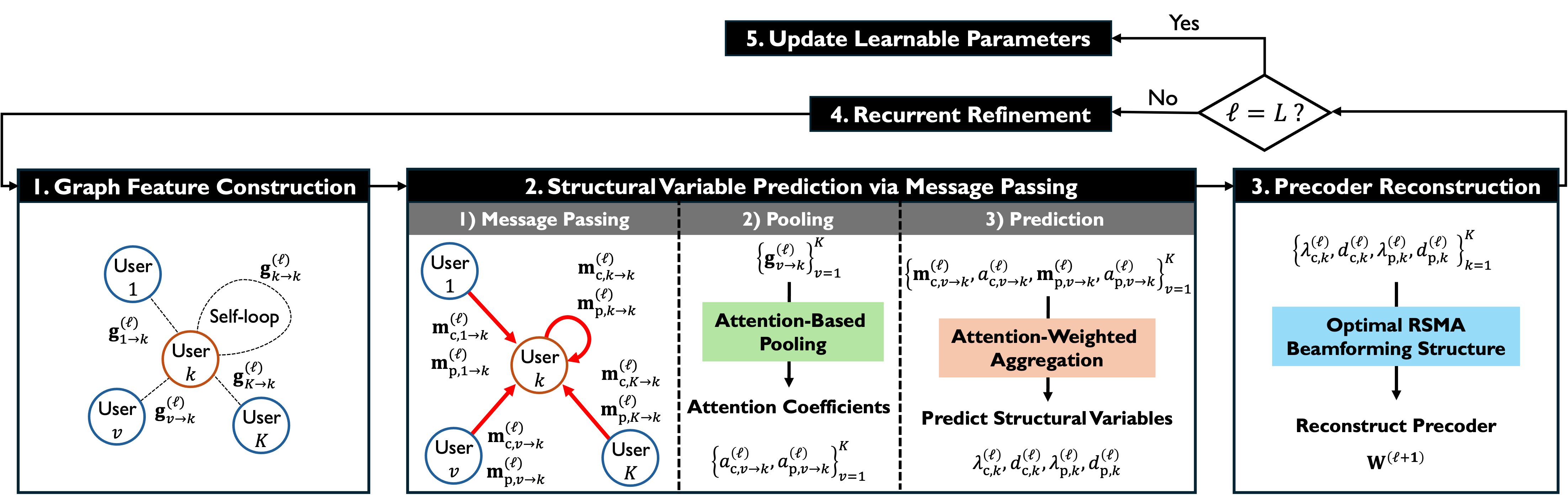}
    \caption{The architecture of the proposed RS-GNN.}
    \label{fig:framework}
\end{figure*}
\subsection{Design Principle of RS-GNN}

We begin by presenting the optimal RSMA beamforming structure that serves as the analytical foundation of RS-GNN, and then describe how RS-GNN exploits this structure to enable closed-loop interference-aware refinement for scalable RSMA precoding.

\subsubsection{Optimal RSMA Beamforming Structure}

Following the Karush--Kuhn--Tucker~(KKT)-based optimal beamforming structure for RSMA in~\cite{fang2024rate}, the common and private precoders at any non-zero stationary point of problem \(\mathcal P\) can be represented as follows:
\begin{align}
\mathbf{w}_{\rm c} &= \mathbf{H}
\big( \mathbf{I}_{K} + \operatorname{diag}(\lambdav_{\rm c})\mathbf{H}^{\rm H}\mathbf{H}
\big)^{-1} \mathbf{d}_{\rm c},\nonumber\\
\mathbf{W}_{\rm p} &=
\mathbf{H} \big( \mathbf{I}_{K}
+ \operatorname{diag}(\lambdav_{\rm p}) \mathbf{H}^{\rm H}\mathbf{H}
\big)^{-1} \operatorname{diag} ( \mathbf{d}_{\rm p}),\label{eq:rsma_structure}
\end{align} 
where $\boldsymbol{\lambda}_{\rm c}=
[\lambda_{{\rm c},1},\lambda_{{\rm c},2},\dots,\lambda_{{\rm c},K}]^{\rm T}
\in\mathbb{R}_{+}^{K}$ and
$\boldsymbol{\lambda}_{\rm p}=
[\lambda_{{\rm p},1},\lambda_{{\rm p},2},\dots,\lambda_{{\rm p},K}]^{\rm T}
\in\mathbb{R}_{+}^{K}$ are non-negative structural
weight vectors, $\mathbf{d}_{\rm c}=
[d_{{\rm c},1},d_{{\rm c},2},\dots,d_{{\rm c},K}]^{\rm T}
\in\mathbb{C}^{K}$ is the channel combination
coefficient vector for the common precoder, and
$\mathbf{d}_{\rm p}=
[d_{{\rm p},1},d_{{\rm p},2},\dots,d_{{\rm p},K}]^{\rm T}
\in\mathbb{R}_{+}^{K}$ is that for the private
precoders. The components of $\mathbf{d}_{\rm p}$ can be restricted to be non-negative without loss of generality, since each $d_{{\rm p},k}$ only scales the \(k\)-th private beam and its phase can be absorbed into the corresponding data symbol. In contrast, $\mathbf{d}_{\rm c}$ is kept complex-valued, as its components jointly determine the channel combination direction of the common precoder.

\subsubsection{Closed-Loop Structural Refinement}

For a given composite channel matrix~$\mathbf{H}$, the RSMA precoding matrix~$\mathbf{W}=[\mathbf{w}_{\rm c}, \mathbf{W}_{\rm p}] \in\mathbb{C}^{N\times(K+1)}$, comprising $2N(K+1)$ real-valued parameters, is fully
determined by $5K$ real-valued user-wise structural variables
$\{\lambda_{{\rm c},k}, \lambda_{{\rm p},k},
{\rm Re}(d_{{\rm c},k}), {\rm Im}(d_{{\rm c},k}),
d_{{\rm p},k}\}_{k=1}^{K}$, reducing the neural network output dimension by a factor of approximately $\frac{2N(K+1)}{5K}\approx\frac{2N}{5}$. A neural network that predicts only these $O(K)$ structural variables is naturally decoupled from the antenna dimension~$N$, and the full precoder is recovered through the optimal structure in~\eqref{eq:rsma_structure}, whose matrix dimensions automatically adapt to the input channel, which is key to size-generalizable precoding.

Building on this structure, RS-GNN couples the message passing dynamics directly with 
the evolving precoding solution. At each of the $L$ refinement layers~$\ell\in\{0,1,\dots,L-1\}$, the current RSMA precoding state $\mathbf{W}^{(\ell)}=\Big[\mathbf{w}_{\rm c}^{(\ell)}, \mathbf{W}_{\rm p}^{(\ell)}\Big]$ is used to reconstruct precoder-dependent graph features, so that the message passing module explicitly observes the pattern of the desired and interference beams induced by the current precoding state. The predicted structural variables are then used to recover the next precoder through the structure-preserving reconstruction based on \eqref{eq:rsma_structure}, closing the feedback loop between successive refinement layers:
\begin{align}
\mathbf W^{(\ell)}
&\stackrel{(a)}{\Rightarrow}
\mathbf g_{v\to k}^{(\ell)}\nonumber\\
&\stackrel{(b)}{\Rightarrow}
\Big\{{
\lambda_{{\rm c},k}^{(\ell)},{\rm Re}\Big(d_{{\rm c},k}^{(\ell)}\Big),{\rm Im}\Big(d_{{\rm c},k}^{(\ell)}\Big),
\lambda_{{\rm p},k}^{(\ell)},
d_{{\rm p},k}^{(\ell)}
}
\Big\}_{k=1}^{K}\nonumber\\
&\stackrel{(c)}{\Rightarrow}
\mathbf W^{(\ell+1)},\label{eq:refinement_loop}
\end{align} where (a), (b), and (c) denote precoder-dependent graph feature construction from the current precoding state $\mathbf{W}^{(\ell)}$ beyond the fixed channel matrix~$\mathbf{H}$, structural variable prediction via message passing that captures the current desired and interference beam patterns, and structure-based precoder reconstruction from the predicted low-dimensional structural variables, respectively. This design decouples the structural variable prediction from fixed system dimensions, enabling size-generalizable inference, which is the key distinction of RS-GNN from existing GNN-based precoding methods.

\subsection{RS-GNN Architecture}
\label{subsec:RS-GNN}

The proposed RS-GNN architecture operates as a recurrent refinement process, as summarized in~\eqref{eq:refinement_loop} and illustrated in Fig.~\ref{fig:framework}. At each layer~$\ell$, the current RSMA precoding state $\mathbf{W}^{(\ell)}$ is updated through three sequential steps: precoder-dependent graph feature construction, structural variable prediction via message passing, and structure-based precoder reconstruction, corresponding to steps~(a),~(b), and~(c) in \eqref{eq:refinement_loop}, respectively.  

\subsubsection{Precoder Initialization}
The precoding state is initialized at layer $\ell=0$ as follows:
\begin{equation}
\mathbf W^{(0)}
= \sqrt{\frac{P_{\rm TX}}{\big\|\tilde{\mathbf W}^{(0)}\big\|_F^2}}\tilde{\mathbf W}^{(0)},\label{eq:init}
\end{equation} where 
$\tilde{\mathbf W}^{(0)} = \big[\sum_{k\in\mathcal K}\mathbf h_k,\mathbf H\big]$. The common precoder is initialized along the aggregate channel direction and the private precoders are initialized by matched-filter transmission. This initialization is required for the permutation equivariance analysis in Section~\ref{subsec:PEA}.

\subsubsection{Precoder-Dependent Graph Feature Construction}
Given the precoding state $\Wm^{(\ell)}$ at layer $\ell$, the edge feature for each directed edge $(v,k)$, $\forall v\in \mathcal{K}\setminus\{k\}$, is defined as
\begin{align}
\mathbf{g}_{v\to k}^{(\ell)}=
\Big[
\big|\mathbf{h}_v^{\rm H}\mathbf{h}_k\big|,
\big|\mathbf{h}_v^{\rm H}\mathbf{w}_{\rm c}^{(\ell)}\big|,&\big|\mathbf{h}_k^{\rm H}\mathbf{w}_{\rm c}^{(\ell)}\big|,\nonumber\\
&\big|\mathbf{h}_v^{\rm H}\mathbf{w}_{{\rm p},k}^{(\ell)}\big|,
\big|\mathbf{h}_k^{\rm H}\mathbf{w}_{{\rm p},v}^{(\ell)}\big|
\Big]^{\rm T}. 
\label{eq:rsma_edgefeature}
\end{align}
The first component measures the channel correlation between neighboring node~$v$ and target node~$k$. The second and third components quantify how the current common precoder is observed by user~$v$ and user~$k$. The fourth and fifth components capture the leakage of private beams toward each other. 

To construct the node graph feature of user~$k$ at layer $\ell$ using shared networks, we augment the directed graph by introducing a self-loop $(k,k)$ for each user node $k$. The self-loop feature is defined by applying the same feature template in \eqref{eq:rsma_edgefeature} with $v=k$, yielding $\mathbf{g}_{k\to k}^{(\ell)}$. Unlike the directed edges, the self-loop does not represent an inter-user interference link; rather, it characterizes the effective common beam gain and the desired private beam gain at user $k$ under the current precoders. The repeated entries in~\eqref{eq:rsma_edgefeature}, i.e., the second and third components and the fourth and fifth components, ensure that the self-loop features can be processed by the shared networks as the inter-user edge features. Since all graph features depend on the current common and private precoders, they are reconstructed at every refinement layer, making the graph representation precoder-dependent rather than fixed by the channel realization alone.

\subsubsection{Structural Variable Prediction
via Message Passing}

Given the graph feature vectors $\mathbf{g}_{v\to k}^{(\ell)},\forall v\in\Kc$ (including the node graph feature), RS-GNN employs two shared feed-forward networks for the common and private branches, respectively. The common branch network $\phi_{\rm comm}^{\rm RS\text {-}GNN}:\mathbb{R}^{5}\to\mathbb{R}^{3}$ generates three messages for estimating the structural variables of user $k$ associated with the common precoder $\wv_{\rm c}^{(\ell)}$:
\begin{align}
    \mathbf m_{{\rm c},v\to k}^{(\ell)}
    &=
    \phi_{\rm comm}^{\rm RS\text {-}GNN}\Big(\mathbf g_{v\to k}^{(\ell)}\Big) \nonumber\\
    &=
    \Big[
    m_{\lambda_{\rm c},v\to k}^{(\ell)},
    m_{\Re(d_{\rm c}),v\to k}^{(\ell)},
    m_{\Im(d_{\rm c}),v\to k}^{(\ell)}
    \Big]^{\rm T},
\end{align}
where $m_{\lambda_{\rm c},v\to k}^{(\ell)}$ is the message for the common-stream structural weight. Since $d_{{\rm c},k}^{(\ell)}$ is complex-valued, its message is represented by two real-valued components $m_{\Re(d_{\rm c}),v\to k}^{(\ell)}$ and $m_{\Im(d_{\rm c}),v\to k}^{(\ell)}$, corresponding to the real and imaginary parts, respectively. The private branch network $\phi_{\rm priv}^{\rm RS\text {-}GNN}:\mathbb{R}^{5}\to\mathbb{R}^{2}$ generates two messages for the private precoders $\Wm_{\rm p}^{(\ell)}$:
\begin{align}
    \mathbf m_{{\rm p},v\to k}^{(\ell)}
    &=
    \phi_{\rm priv}^{\rm RS\text {-}GNN}\Big(\mathbf g_{v\to k}^{(\ell)}\Big)\nonumber\\
    &=
    \Big[
    m_{\lambda_{\rm p},v\to k}^{(\ell)},
    m_{d_{\rm p},v\to k}^{(\ell)}
    \Big]^{\rm T},
\end{align}
where $m_{\lambda_{\rm p},v\to k}^{(\ell)}$ and $m_{d_{\rm p},v\to k}^{(\ell)}$ are the messages for the private-stream structural weight and channel combination coefficient, respectively.

To weight the incoming messages according to their relevance to the common and private structural variables, we adopt branch-specific attention-based pooling. The common and private branch attention scores are computed as
\begin{align}
{\sigma}_{{\rm c},v\to k}^{(\ell)}
&=
\operatorname{LeakyReLU}
\Big(
\psi_{\rm comm}^{\rm RS\text {-}GNN}\Big(\mathbf g_{v\to k}^{(\ell)}\Big)
\Big),\nonumber\\
\sigma_{{\rm p},v\to k}^{(\ell)}
&=
\operatorname{LeakyReLU}
\Big(
\psi_{\rm priv}^{\rm RS\text {-}GNN}\Big(\mathbf g_{v\to k}^{(\ell)}\Big)\Big),
\end{align}
where $\psi_{\rm comm}^{\rm RS\text {-}GNN}:\mathbb{R}^{5}\to\mathbb{R}$ and $\psi_{\rm priv}^{\rm RS\text {-}GNN}:\mathbb{R}^{5}\to\mathbb{R}$ are shared linear scoring functions. From the attention scores, the attention coefficients are obtained as
\begin{align}
{a}_{{\rm c},v\to k}^{(\ell)}
&=\frac{
\exp\Big(\sigma_{{\rm c},v\to k}^{(\ell)}\Big)
}{
\sum_{i\in\mathcal{K}}
\exp\Big(\sigma_{{\rm c},i\to k}^{(\ell)}\Big)
},\nonumber\\
a_{{\rm p},v\to k}^{(\ell)}
&=\frac{
\exp\Big(\sigma_{{\rm p},v\to k}^{(\ell)}\Big)
}{
\sum_{i\in \mathcal{K}}
\exp\Big(\sigma_{{\rm p},i\to k}^{(\ell)}\Big)
}.
\end{align} The user-wise structural variables are then predicted by attention-weighted aggregation:
\begin{align}
\lambda_{{\rm c},k}^{(\ell)}
&=
\operatorname{softplus}
\Bigg(
\sum_{v\in\mathcal{K}}
a_{{\rm c},v\to k}^{(\ell)}
m_{\lambda_{\rm c},v\to k}^{(\ell)}
\Bigg),\nonumber\\
d_{{\rm c},k}^{(\ell)}
&=
\sum_{v\in\mathcal{K}}
a_{{\rm c},v\to k}^{(\ell)}\Big(m_{\Re(d_{\rm c}),v\to k}^{(\ell)}+{j}  m_{\Im(d_{\rm c}),v\to k}^{(\ell)}\Big)
,\nonumber\\
\lambda_{{\rm p},k}^{(\ell)}
&=
\operatorname{softplus}
\Bigg(
\sum_{v\in\mathcal{K}}
a_{{\rm p},v\to k}^{(\ell)}
m_{\lambda_{\rm p},v\to k}^{(\ell)}
\Bigg),\nonumber\\
d_{{\rm p},k}^{(\ell)}
&=
\operatorname{softplus}
\Bigg(
\sum_{v\in\mathcal{K}}
a_{{\rm p},v\to k}^{(\ell)}
m_{d_{\rm p},v\to k}^{(\ell)}
\Bigg).
\end{align} The softplus activation enforces the non-negativity of $\lambda_{{\rm c},k}^{(\ell)}$, $\lambda_{{\rm p},k}^{(\ell)}$, and $d_{{\rm p},k}^{(\ell)}$, as required by the optimal beamforming structure in~\eqref{eq:rsma_structure}.

\subsubsection{Precoder Reconstruction}

After message passing, the predicted structural variables are used to update the RSMA precoder through the optimal structure in~\eqref{eq:rsma_structure}, replacing the generic combination function~\eqref{eq:combination}. Specifically, let
\begin{align}
\mathbf{C}_{\rm c}^{(\ell)}&=\mathbf{I}_K+\operatorname{diag}\Big(\lambdav_{\rm c}^{(\ell)}\Big)
\mathbf{H}^{\rm H}\mathbf{H},\nonumber\\
\mathbf{C}_{\rm p}^{(\ell)}&=\mathbf{I}_K+\operatorname{diag}\Big(\lambdav_{\rm p}^{(\ell)}\Big)\mathbf{H}^{\rm H}\mathbf{H}.
\label{eq:C}
\end{align} Rather than inverting the $N\times N$ antenna-dimensional matrices directly, the reconstruction is equivalently performed by solving two $K\times K$ linear systems:
\begin{align}
\mathbf{C}_{\rm c}^{(\ell)}\mathbf{x}_{\rm c}^{(\ell)}&=\mathbf{d}_{\rm c}^{(\ell)},\nonumber\\
\mathbf{C}_{\rm p}^{(\ell)}\mathbf{X}_{\rm p}^{(\ell)} &=\operatorname{diag}\Big(\mathbf{d}_{\rm p}^{(\ell)}\Big),
\label{eq:linearsolve}
\end{align}
where $\mathbf{x}_{\rm c}^{(\ell)}\in\mathbb{C}^{K}$ and $\mathbf{X}_{\rm p}^{(\ell)}\in\mathbb{C}^{K\times K}$ are the common- and private-stream coefficient vector and matrix, respectively. Although the Gram matrix $\mathbf{H}^{\mathrm H}\mathbf H$ can be rank-deficient when $K>N$ or ill-conditioned under highly correlated user channels, the coefficient matrices $\mathbf C_c^{(\ell)}$ and $\mathbf C_p^{(\ell)}$ in~\eqref{eq:C} remain non-singular due to the identity term and the non-negative structural weights. Hence, the linear systems in~\eqref{eq:linearsolve} each admit a unique solution even in overloaded or highly correlated channel regimes. In practice, the systems are solved using an LU decomposition-based linear solver. Letting $\hat{\mathbf{x}}_{\rm c}^{(\ell)}\in\CC^{K}$ and $\hat{\mathbf{X}}_{\rm p}^{(\ell)}\in\CC^{K\times K}$ denote the estimated coefficient vector and matrix from the linear systems, the unnormalized RSMA precoder is then reconstructed as
\begin{equation}
\tilde{\mathbf{W}}^{(\ell+1)}=\Big[\tilde{\mathbf{w}}_{\rm c}^{(\ell+1)},\tilde{\mathbf{W}}_{\rm p}^{(\ell+1)}\Big]=\mathbf{H}\Big[\hat{\mathbf{x}}_{\rm c}^{(\ell)},\hat{\mathbf{X}}_{\rm p}^{(\ell)}\Big],
\label{eq:reconstruction}
\end{equation}
and the transmit power constraint is enforced by normalizing:
\begin{equation}
\mathbf{W}^{(\ell+1)}=\sqrt{\frac{P_{\rm TX}}{\big\|\tilde{\mathbf{W}}^{(\ell+1)}\big\|_F^2}}\,\tilde{\mathbf{W}}^{(\ell+1)}.
\label{eq:powernormalization}
\end{equation}

Since the graph feature construction, message passing, the linear systems, and power normalization are all differentiable, the entire recurrent refinement process supports end-to-end training via backpropagation of the sum-rate loss through all refinement layers, as described in Section~\ref{sec:TrainInfer}.

\subsection{Permutation Equivariance Analysis}
\label{subsec:PEA}

A desirable property of any size-generalizable precoding architecture is that its output should be independent of the arbitrary indexing of users and antennas. For standard GNNs, user permutation equivariance follows directly from parameter sharing and permutation-invariant pooling. However, verifying this property for RS-GNN is non-trivial because the architecture incorporates additional structure-preserving operations (precoder-dependent graph feature construction, linear systems, and power normalization), each of which must be shown to respect the permutation structure. Moreover, RSMA imposes a {\em mixed permutation equivariance} that differs from conventional SDMA precoding: the common precoder must be invariant to user reordering, whereas the private precoders must be equivariant, as stated in the following theorem.

\begin{theorem}[Mixed Permutation Equivariance of RS-GNN]\label{thm:mixedPE}
Let $\boldsymbol{\Pi}_{\rm A}\in\{0,1\}^{N\times N}$ and $\boldsymbol{\Pi}_{\rm U}\in\{0,1\}^{K\times K}$ denote arbitrary antenna and user permutation matrices, and define the permuted matrix of a composite channel matrix $\Hm$ as
\begin{equation}
  \bar{\mathbf{H}}=\boldsymbol{\Pi}_{\rm A}\mathbf{H}\boldsymbol{\Pi}_{\rm U}.
\end{equation}
Let $\mathbf{W}^{(\ell)}=\Big[\mathbf{w}_{\rm c}^{(\ell)},\mathbf{W}_{\rm p}^{(\ell)}\Big]$ and $\bar{\mathbf{W}}^{(\ell)}=\Big[\bar{\mathbf{w}}_{\rm c}^{(\ell)}, \bar{\mathbf{W}}_{\rm p}^{(\ell)}\Big]$ denote the layer-$\ell$ outputs of RS-GNN for inputs $\mathbf{H}$ and $\bar{\mathbf{H}}$, respectively. Then, the proposed RS-GNN satisfies
\begin{align}
  \bar{\mathbf{w}}_{\rm c}^{(\ell)}&=\boldsymbol{\Pi}_{\rm A}\mathbf{w}_{\rm c}^{(\ell)},\\
  \bar{\mathbf{W}}_{\rm p}^{(\ell)}&=\boldsymbol{\Pi}_{\rm A}\mathbf{W}_{\rm p}^{(\ell)}\boldsymbol{\Pi}_{\rm U},\forall\ell.
\end{align} 
\end{theorem}
\begin{IEEEproof}
The proof is provided in Appendix~\ref{sec:apxA}.
\end{IEEEproof}
\vspace{0.1cm}

Although RS-GNN is designed for RSMA precoding, it can be readily reduced to conventional SDMA precoding by deactivating the common-stream branch ($\mathbf{w}_{\rm c}
=\mathbf{0}_N$), retaining the same message passing module, $K\times K$ linear systems, and power normalization without architectural modification.

\begin{corollary}[Permutation Equivariance of RS-GNN in SDMA] \label{cor:sdma}
When the common-stream is deactivated, RS-GNN reduces to an SDMA precoder that preserves the user and antenna permutation equivariance:
\begin{equation}
\bar{\mathbf W}_{\rm p}^{(\ell)}
=
\boldsymbol\Pi_{\rm A}
\mathbf W_{\rm p}^{(\ell)}
\boldsymbol\Pi_{\rm U},\forall \ell .
\end{equation}
\end{corollary}
\begin{IEEEproof}
This follows directly from Theorem~\ref{thm:mixedPE} by setting $\mathbf{w}_{\rm c}=\mathbf{0}_N$.
\end{IEEEproof}

The converse, however, cannot be directly applied without architectural modification for existing GNN-based SDMA precoders~\cite{shen2022graph, guo2023model,zhang2025gradient, sun2025homogeneous, chowdhury2023deep}, which are designed to output only user-indexed private precoders. As established in Theorem~\ref{thm:mixedPE}, RSMA precoding requires a mixed permutation structure in which the common precoder must be user-invariant while the private precoders must be user-equivariant. Consequently, simply replacing the SDMA objective with the RSMA sum-rate is insufficient without an explicit mechanism for generating the user-invariant common precoder. The SDMA-reduced version of RS-GNN, referred to as {\bf RS-GNN-SDMA}, is evaluated against existing GNN-based SDMA precoders in Section~\ref{subsec:SDMA}, demonstrating that the proposed structure-preserving design generalizes beyond RSMA and remains competitive in the conventional SDMA setting without any architectural modification.

\section{Training and Inference of RS-GNN}
\label{sec:TrainInfer}

In this section, we describe how RS-GNN is trained and deployed for inference. The training procedure is unsupervised and relies solely on the RSMA sum-rate objective, while the inference procedure supports variable-depth execution and generalizes to unseen system sizes without retraining.

\subsection{Unsupervised Training with Cumulative Sum-Rate Loss}
\label{subsec:training}

Let {$\boldsymbol{\theta}^{\rm RS\text {-}GNN}\in\mathbb R^{N_{\theta}}$} denote all learnable parameters of RS-GNN, obtained by vectorizing the parameters of the common and private branch networks $\phi_{\rm comm}^{\rm RS\text {-}GNN}(\cdot)$ and $\phi_{\rm priv}^{\rm RS\text {-}GNN}(\cdot)$, and the attention 
modules $\psi_{\rm comm}^{\rm RS\text {-}GNN}(\cdot)$ and $\psi_{\rm priv}^{\rm RS\text {-}GNN}(\cdot)$. The whole network is trained in an unsupervised manner by directly maximizing the RSMA sum-rate in \eqref{eq:rsum} over a mini-batch of channel realizations~$\mathcal{B}$, without requiring labeled outputs from any numerical solver.

To make every intermediate recurrent output an effective RSMA precoder, we employ a negative cumulative sum-rate loss over all refinement steps jointly:
\begin{equation}
  \mathcal{L}\Big(\boldsymbol{\theta}^{\rm RS\text {-}GNN}\Big)
  =-\frac{1}{|\mathcal{B}|}
  \sum_{\mathbf{H}\in\mathcal{B}}
  \sum_{\ell=1}^{L_{\rm train}}
  \omega_{\ell}\,R_{\rm sum}\Big(\mathbf{W}^{(\ell)} ; \mathbf H\Big),
  \label{eq:loss}
\end{equation}
where $L_{\rm train}$ denotes the training refinement depth and $\omega_\ell\geq 0$ is the corresponding loss weight. Non-decreasing weights $\omega_1\leq\omega_2\leq \cdots\leq\omega_{L_{\rm train}}$ are adopted so that later refinement outputs are emphasized while intermediate outputs are still explicitly supervised throughout training. For notational clarity, we write $R_{\rm sum}(\mathbf{W}^{(\ell)};\mathbf{H}) \triangleq R_{\rm sum}(\mathbf{W}^{(\ell)})$, where $R_{\rm sum}(\cdot)$ is defined in~\eqref{eq:rsum} and the semicolon explicitly indicates the dependence on the channel realization~$\mathbf{H}$. Since the loss in~\eqref{eq:loss} depends only on the achievable RSMA rates, the structural variables are learned implicitly through sum-rate maximization, without any additional supervision. 

As the graph feature construction, message passing, $K\times K$ linear systems, and power normalization are all differentiable, the entire recurrent refinement process supports end-to-end training via {gradient-based stochastic optimization for minimizing the sum-rate loss in \eqref{eq:loss}} through all 
refinement layers. Consequently, RS-GNN learns to predict the structural variables directly from the sum-rate objective while preserving the analytical beamforming structure throughout the forward pass. This training procedure is summarized in Algorithm~\ref{alg:train_rsgnn}.

\subsection{Size-Generalizable Inference}
\label{subsec:SGinference}

After offline training, the trained parameter vector denoted as $\tilde{\boldsymbol{\theta}}^{\rm RS\text {-}GNN}$ is fixed and used for online precoder generation. For a given channel realization, RS-GNN  starts from the initialized precoding state $\mathbf W^{(0)}$ and executes the recurrent refinement procedure for a prescribed inference depth 
$L_{\rm inf}$. Specifically, at each $\ell=0,\ldots, L_{\rm inf}-1$, the model constructs precoder-dependent graph features from the current precoding state $\mathbf W^{(\ell)}$, predicts the structural variables using the trained shared networks, and reconstructs the next precoder $\mathbf W^{(\ell+1)}$ through the linear systems in \eqref{eq:C} and \eqref{eq:linearsolve}. The final output $\mathbf W^{(L_{\rm inf})}$ is used as the RSMA precoder denoted as $\hat{\Wm}$. 
Since the cumulative loss in~\eqref{eq:loss} supervises every intermediate output, $L_{\rm inf}$ can be set to any value satisfying $L_{\rm inf}\leq L_{\rm train}$ without retraining, making it a controllable knob for the online computational budget: a smaller $L_{\rm inf}$ reduces inference time, whereas a larger $L_{\rm inf}$ yields a more refined precoding solution.

RS-GNN can further be directly applied to channel matrices of unseen system size $(N',K')$ without retraining or architectural modification. This generalization capability stems from two structural properties of the architecture. The learnable processing networks and attention module are shared across all user nodes and edges, operating on fixed-dimensional edge features regardless of the system size. The structure-based reconstruction in~\eqref{eq:C} and~\eqref{eq:linearsolve} automatically adapts its 
matrix dimensions to the input channel, since the $K\times K$ linear systems scale with the number of users rather than the number of antennas. Specifically, for a new composite channel size $\mathbf{H}\in\mathbb{C}^{N'\times K'}$, RS-GNN predicts structural variables $\boldsymbol{\lambda}_{\rm c},\boldsymbol{\lambda}_{\rm p}, \mathbf{d}_{\rm p}\in\mathbb{R}^{K'}$ and $\mathbf{d}_{\rm c}\in\mathbb{C}^{K'}$, and reconstructs a precoding matrix $\mathbf{W}^{(L_{\rm inf})}\in \mathbb{C}^{N'\times(K'+1)}$ using the same trained parameters~$\tilde{\boldsymbol{\theta}}^{\rm RS\text {-}GNN}$. The trained model is therefore not tied to the training configuration $(N,K)$ and generalizes to arbitrary unseen system sizes $(N',K')$, while preserving the mixed permutation equivariance established in Theorem~\ref{thm:mixedPE}. The inference procedure is summarized in Algorithm~\ref{alg:infer_rsgnn}.

\begin{algorithm}[t]
\caption{Training Procedure of RS-GNN}
\label{alg:train_rsgnn}
\begin{algorithmic}[1]
\STATE \textbf{Input:} Mini-batch of channel realizations $\mathcal B$, power budget $P_{\rm TX}$, training refinement depth $L_{\rm train}$, loss weights $\{\omega_\ell\}_{\ell=1}^{L_{\rm train}}$

\STATE \textbf{Initialization:}  $\boldsymbol{\theta}^{\rm RS\text{-}GNN}$, {$\mathbf W^{(0)}$ (defined in~\eqref{eq:init})}
\FOR{$\ell=0,\ldots, L_{\rm train}-1$}
    \STATE Construct precoder-dependent graph features from \eqref{eq:rsma_edgefeature}
    \STATE Predict $\boldsymbol{\lambda}_{\rm c}^{(\ell)}$, $\boldsymbol{\lambda}_{\rm p}^{(\ell)}$, $\mathbf d_{\rm c}^{(\ell)}$, and $\mathbf d_{\rm p}^{(\ell)}$ via message passing
    \STATE Form $\mathbf C_{\rm c}^{(\ell)}$ and $\mathbf C_{\rm p}^{(\ell)}$ using \eqref{eq:C}
    \STATE Obtain $\hat{\mathbf x}_{\rm c}^{(\ell)}$ and $\hat{\mathbf X}_{\rm p}^{(\ell)}$ from the linear systems in \eqref{eq:linearsolve}
    \STATE Reconstruct $\tilde{\mathbf W}^{(\ell+1)}$ using \eqref{eq:reconstruction}
    \STATE Normalize $\tilde{\mathbf W}^{(\ell+1)}$ using \eqref{eq:powernormalization} to obtain $\mathbf W^{(\ell+1)}$
\ENDFOR
\STATE Compute the cumulative loss \(\mathcal L(\boldsymbol{\theta}^{\rm RS\text{-}GNN})\) in \eqref{eq:loss}
\STATE Update $\boldsymbol{\theta}^{\rm RS\text{-}GNN}$ via gradient-based stochastic optimization
\STATE \textbf{Output:} Trained parameters $\tilde{\boldsymbol{\theta}}^{\rm RS\text{-}GNN}$
\end{algorithmic}
\end{algorithm}

\begin{algorithm}[t]
\caption{Inference Procedure of RS-GNN}
\label{alg:infer_rsgnn}
\begin{algorithmic}[1]
\STATE \textbf{Input:} Composite channel matrix $\mathbf H\in\mathbb C^{N'\times K'}$, power budget $P_{\rm TX}$, trained parameters $\tilde{\boldsymbol{\theta}}^{\rm RS\text{-}GNN}$, inference depth $L_{\rm inf}$
\STATE \textbf{Initialization:} $\mathbf W^{(0)}$ (defined in~\eqref{eq:init})

\FOR{$\ell=0,\ldots,L_{\rm inf}-1$}
    \STATE Construct precoder-dependent graph features from \eqref{eq:rsma_edgefeature}
    \STATE Predict $\boldsymbol{\lambda}_{\rm c}^{(\ell)}$, $\boldsymbol{\lambda}_{\rm p}^{(\ell)}$, $\mathbf d_{\rm c}^{(\ell)}$, and $\mathbf d_{\rm p}^{(\ell)}$ from $\tilde{\boldsymbol{\theta}}^{\rm RS\text{-}GNN}$
    \STATE Form $\mathbf C_{\rm c}^{(\ell)}$ and $\mathbf C_{\rm p}^{(\ell)}$ using \eqref{eq:C}
    \STATE Obtain $\hat{\mathbf x}_{\rm c}^{(\ell)}$ and $\hat{\mathbf X}_{\rm p}^{(\ell)}$ from the linear systems in \eqref{eq:linearsolve}
    \STATE Reconstruct $\tilde{\mathbf W}^{(\ell+1)}$ using \eqref{eq:reconstruction}
    \STATE Normalize $\tilde{\mathbf W}^{(\ell+1)}$ using \eqref{eq:powernormalization} to obtain $\mathbf W^{(\ell+1)}$
\ENDFOR
\STATE \textbf{Output:} RSMA precoder $\hat{\mathbf {W}}=\Big[\mathbf w_{\rm c}^{(L_{\rm inf})},\mathbf W_{\rm p}^{(L_{\rm inf})}\Big]$
\end{algorithmic}
\end{algorithm}

\section{Simulation Results}
\label{sec:experiments}

We describe the simulation setup, benchmark methods, and experimental results for the proposed RS-GNN.

\subsection{Simulation Setup}\label{subsec:setup}

The BS is equipped with $N=8$ transmit antennas and serves $K=6$ single-antenna users. The noise variance is normalized to $\sigma_k^2=1$ for all $k$, and the transmit power budget is set to $P_{\rm TX}=10^{\mathrm{SNR}/10}$. The default SNR is $15\;$dB, and perfect CSI is assumed at the BS. These default settings are used throughout unless explicitly stated otherwise.

The channels are generated according to a correlated Rician fading model. Specifically, a cluster center angle is sampled uniformly and the user angles are drawn within a prescribed angular spread around the cluster center. The channel vector of user~$k$ is modeled as
\begin{equation}
\mathbf h_k = \sqrt{\frac{\rho}{\rho+1}} \mathbf a(\theta_k)+
\sqrt{\frac{1}{\rho+1}} \mathbf h_{k,\mathrm{NLoS}},
\end{equation} where $\rho=10$ is the Rician factor, $\mathbf a(\theta_k)$ denotes the uniform linear array~(ULA) steering vector with half-wavelength antenna spacing, and $\mathbf h_{k,\mathrm{NLoS}}\sim\mathcal{CN}(\mathbf 0_{N},\mathbf I_N)$. Unless otherwise specified, the angular spread is set to $30^\circ$ to induce non-negligible correlation among user channel directions, which provides a meaningful test environment for evaluating RSMA interference management~\cite{mao2018rate,hua2023learning}.

In the proposed RS-GNN, the common-stream and private-stream processing networks are implemented as single-hidden-layer feed-forward networks, $\phi_{\rm comm}^{\rm RS\text{-}GNN}:
\mathbb{R}^{5}\to\mathbb{R}^{D_{\rm GNN}}
\to\mathbb{R}^{3}$, $\phi_{\rm priv}^{\rm RS\text{-}GNN}:
\mathbb{R}^{5}\to\mathbb{R}^{D_{\rm GNN}}
\to\mathbb{R}^{2}$, with a LeakyReLU activation after the first linear layer. The attention scoring modules are shared linear maps $\psi_{\rm comm}^{\rm RS\text{-}GNN}, \psi_{\rm priv}^{\rm RS\text{-}GNN}: \mathbb{R}^{5}\to\mathbb{R}$, followed by a LeakyReLU activation before softmax normalization. In all experiments, the hidden dimension and the number of recurrent refinement steps are set to \(D_{\rm GNN}=16\) and \(L_{\rm train}=5\), respectively. The model is trained using $3{,}000$ independent channel realizations with mini-batch size $64$. The learnable parameters are optimized by Adam for $300$ epochs with learning rate $10^{-3}$. The cumulative sum-rate loss in~\eqref{eq:loss} uses uniform weights $\omega_\ell=1/L_{\rm train}$ for all $\ell=1,\ldots,L_{\rm train}$. During training, the common-stream bottleneck rate is approximated by a differentiable soft-min function with temperature $\tau=0.1$, while the hard minimum is used at evaluation time.

The benchmark methods are summarized as follows. We compare the proposed RS-GNN with three representative RSMA precoding methods: RSMA-WMMSE, FP-HFPI, and RS-BNN-U. Since several benchmarks exploit the RSMA optimal-structure-based reconstruction, we also briefly discuss their representative algebraic costs, and the resulting performance--online inference-time tradeoff is then evaluated in Fig.~\ref{Fig:5}.
Regarding the structure-based methods, the dominant algebraic cost stems from the reconstruction of the precoders. Specifically, forming the Gram-related matrices and reconstructing the precoder require $O(NK^2)$ operations, while solving the resulting $K\times K$ linear systems requires $O(K^3)$. Hence, FP-HFPI, RS-BNN-U, and RS-GNN share the structure-based reconstruction cost $O(NK^2+K^3)$. 


\begin{itemize}
\item \textbf{RSMA-WMMSE~\cite{mao2018rate}} is a widely used iterative algorithm for RSMA sum-rate maximization that converges to a stationary solution of the non-convex problem $\mathcal P$. Due to its strong empirical performance, RSMA-WMMSE serves as the
primary high-complexity reference benchmark throughout the simulations.
Its complexity mainly stems from repeatedly solving quadratic precoding subproblems via interior-point methods, with a representative complexity of $O\big(I_{\mathrm W}(KN)^{3.5}\big)$, where $I_{\mathrm W}$ denotes the number of WMMSE iterations. RSMA-WMMSE is run with a maximum of $I_{\mathrm W}=30$ iterations and relative convergence tolerance $10^{-3}$.

\item \textbf{FP-HFPI~\cite{fang2024rate}} is a low-complexity numerical algorithm for RSMA precoding based on fractional programming~(FP) and hyperplane fixed-point iteration~(HFPI). At each outer FP iteration, the FP auxiliary variables are updated and an inner HFPI loop searches for the Lagrangian dual variables. Since each HFPI update invokes the structure-based reconstruction, its representative online complexity scales as $O\big(I_{\rm FP}I_{\rm HFPI}(NK^2+K^3)\big)$, where $I_{\rm FP}$ and $I_{\rm HFPI}$ denote the number of outer FP iterations and the average number of inner HFPI updates per outer iteration, respectively. We set \(I_{\rm FP}=30\) outer FP iterations and \(I_{\rm HFPI}=100\) inner HFPI updates per outer iteration, with outer and inner convergence tolerances \(10^{-4}\) and \(10^{-3}\), respectively.

\item \textbf{RS-BNN-U~\cite{wang2024rs}} is an unsupervised variant of the RS-BNN that unfolds FP-HFPI by replacing its iterative dual-variable update with a dense neural network. Unlike the original supervised-pretraining-plus-unsupervised-fine-tuning protocol, RS-BNN-U is trained solely with the RSMA sum-rate loss for a fair label-free comparison with RS-GNN. Each unfolded layer consists of a dense neural update for predicting the structural variables, followed by the same reconstruction. Thus, its representative online complexity is 
$O\big(L_{\rm BNN}(D_{\rm BNN}NK+NK^2+K^3)\big)$, where $L_{\rm BNN}$ and $D_{\rm BNN}$ denote number of unfolded layers and the hidden width of the dense network, respectively. 
We set \(L_{\rm BNN}=5\) and \(D_{\rm BNN}=512\), and RS-BNN-U is trained with the same training data, optimizer, mini-batch size, learning rate, and number of epochs as RS-GNN.


\item \textbf{RS-GNN} (proposed) predicts the structural variables through precoder-dependent graph message passing over the fully connected user graph, with the representative complexity of $O\big(L_{\rm inf}(D_{\rm GNN}K^2+NK^2+K^3)\big)$. Unlike RS-BNN-U, varying \(L_{\rm inf}\) changes only the number of recurrent executions of the same trained modules.

\end{itemize}

All simulations were conducted on Python~3.11.15. Learning-based methods were implemented in PyTorch 2.12.0 with CUDA 12.8 and evaluated on an NVIDIA GeForce RTX 5080 GPU. For a fair comparison of online computational latency, all learning-based and
iterative numerical methods were executed on an AMD Ryzen 9 9900X CPU with 12 cores using a single channel realization at a time, excluding offline training time for GPU-based methods.


\begin{figure}[t]
  \centering
  \includegraphics[width=1.00\linewidth]{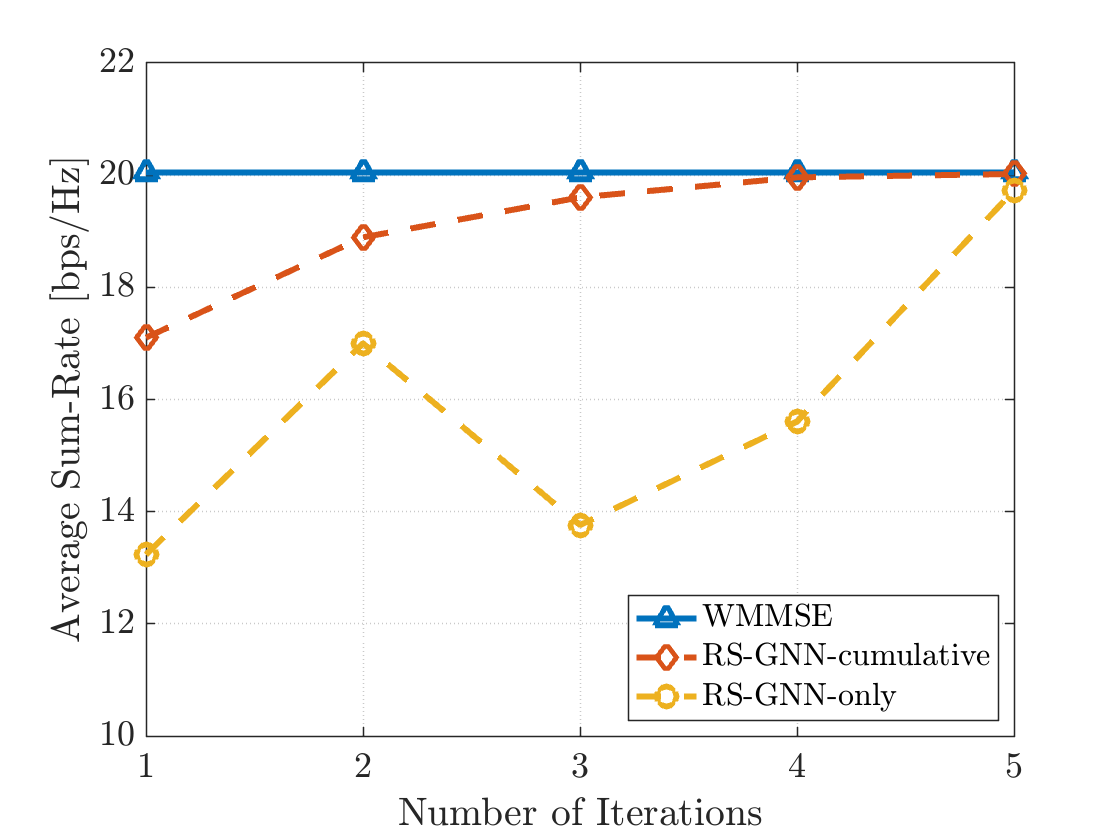}
  \caption{Effect of cumulative training on variable-depth inference with $L_{\rm train}=5$.}
  \label{Fig:3}
\end{figure}

\begin{figure*}[t]
    \centering
    \begin{subfigure}[t]{0.32\textwidth}
        \centering
        \includegraphics[width=\linewidth]{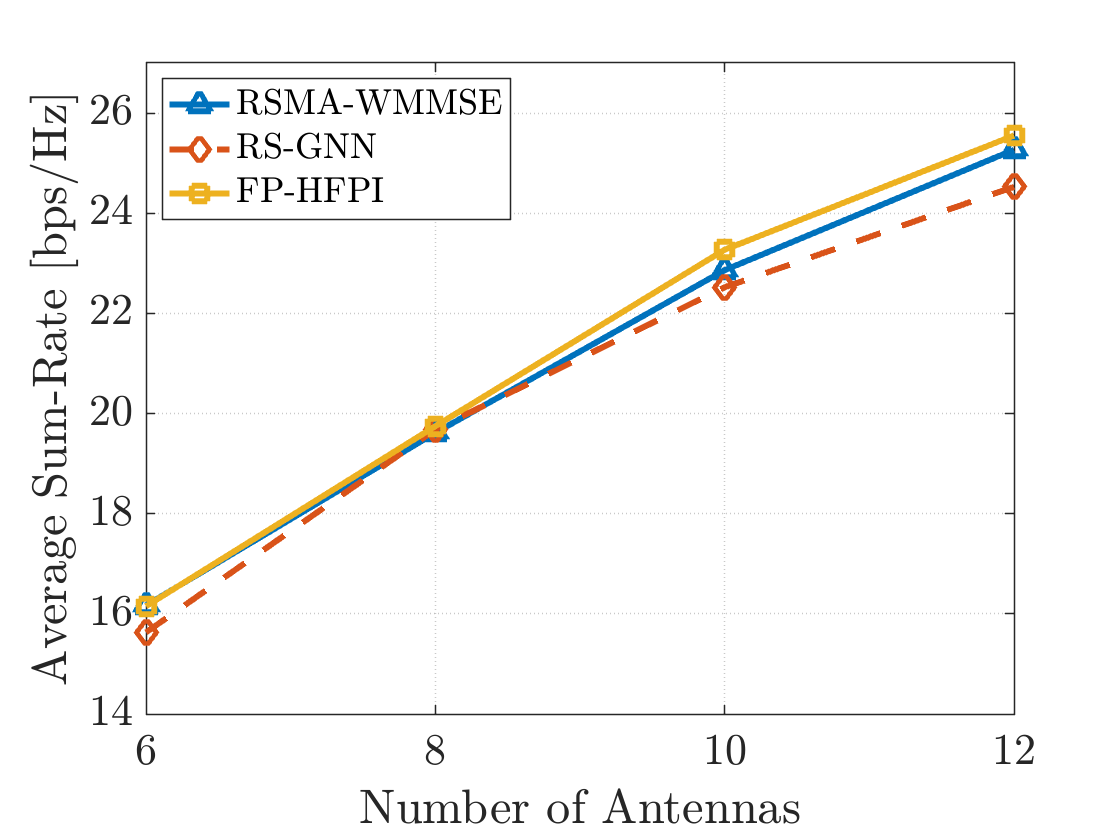}
        \caption{Antenna generalization.}
        \label{fig:4(a)}
    \end{subfigure}
    \hfill
    \begin{subfigure}[t]{0.32\textwidth}
        \centering
        \includegraphics[width=\linewidth]{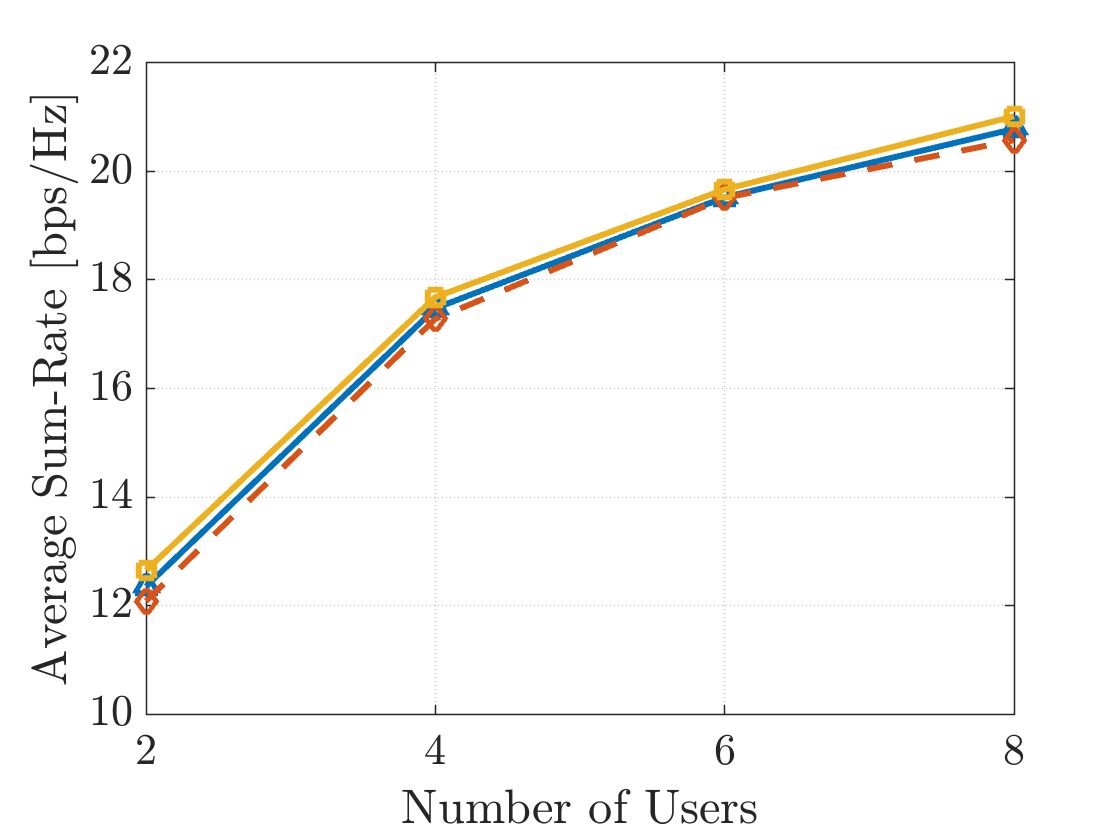}
        \caption{User generalization.}
        \label{fig:4(b)}
        
    \end{subfigure}
    \hfill
    \begin{subfigure}[t]{0.32\textwidth}
        \centering
        \includegraphics[width=\linewidth]{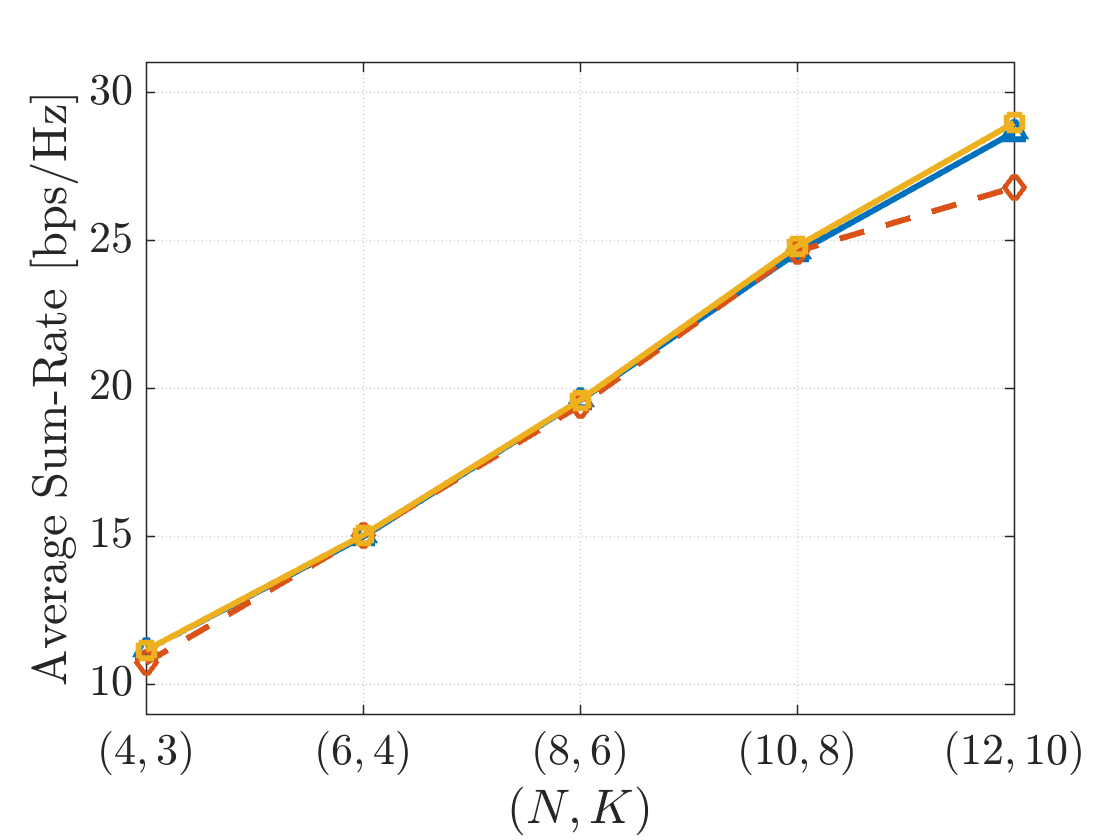}
        \caption{Joint $(N,K)$ generalization.}
        \label{fig:4(c)}
    \end{subfigure}
    \caption{Size generalization of RS-GNN. A single RS-GNN model trained at $(N,K)=(8,6)$ is evaluated without retraining under (a) antenna-size sweep with fixed $K=6$, (b) user-size sweep with fixed $N=8$, and (c) joint $(N,K)$ sweep.}
    \label{fig:4}
\end{figure*}

\begin{figure}[t]
  \centering
  \includegraphics[width=1.00\linewidth]{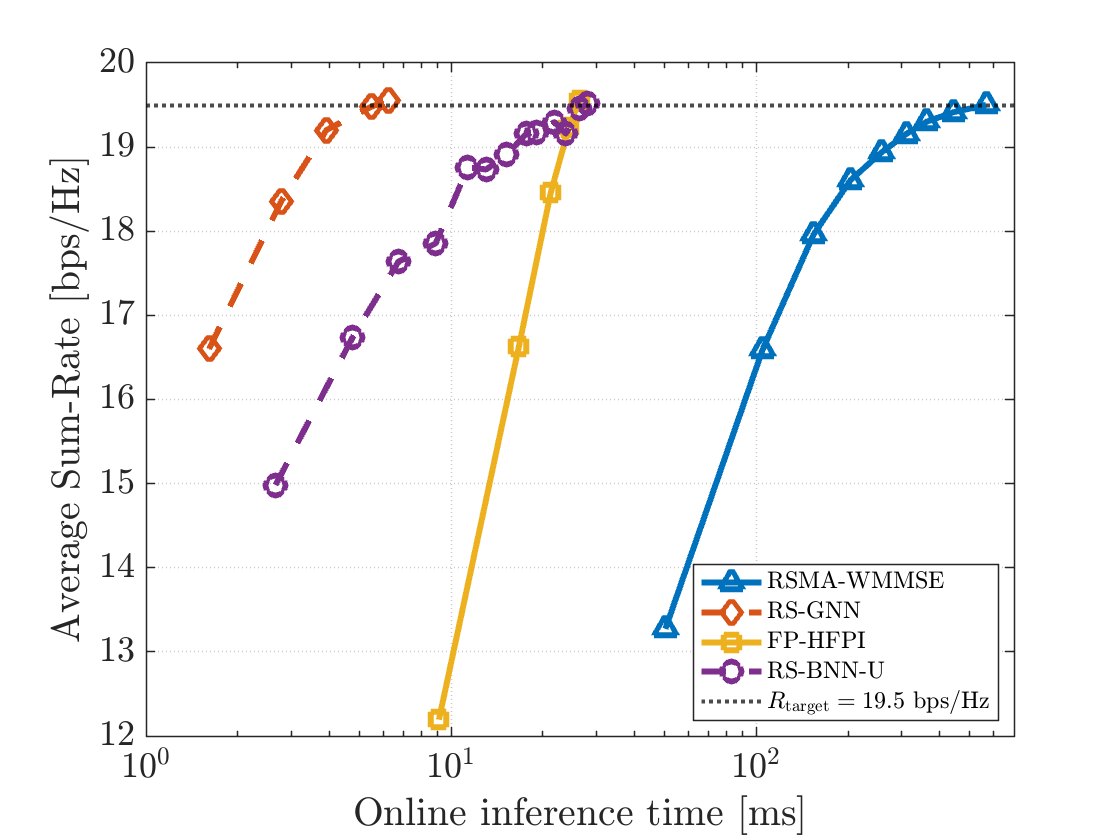}
  \caption{Performance--online-inference-time tradeoff of RSMA precoding methods.}
  \label{Fig:5}
\end{figure}

\begin{figure}[t]
  \centering
  \includegraphics[width=1.00\linewidth]{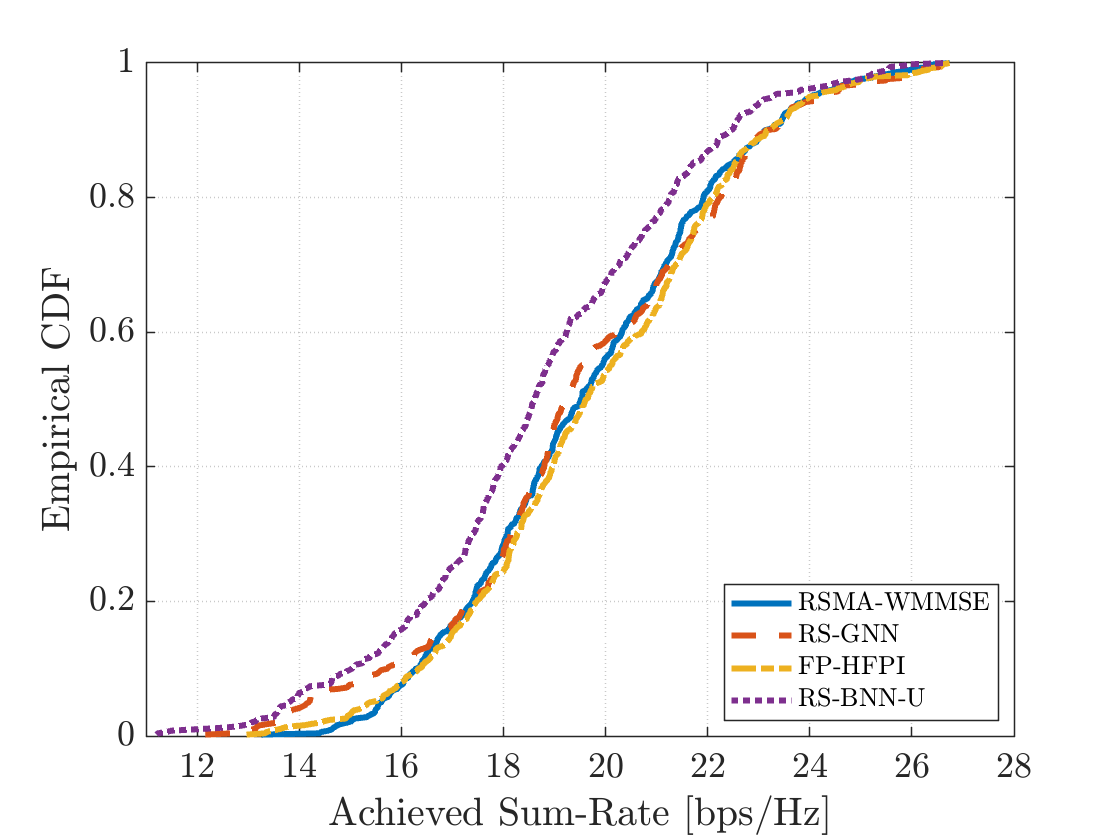}
  \caption{Empirical CDF of the achieved sum-rate over independent test channel realizations.}
  \label{Fig:6}
\end{figure}

\begin{figure*}[!t]
    \centering
    \begin{subfigure}[t]{0.50\textwidth}
        \centering
        \includegraphics[width=\linewidth]{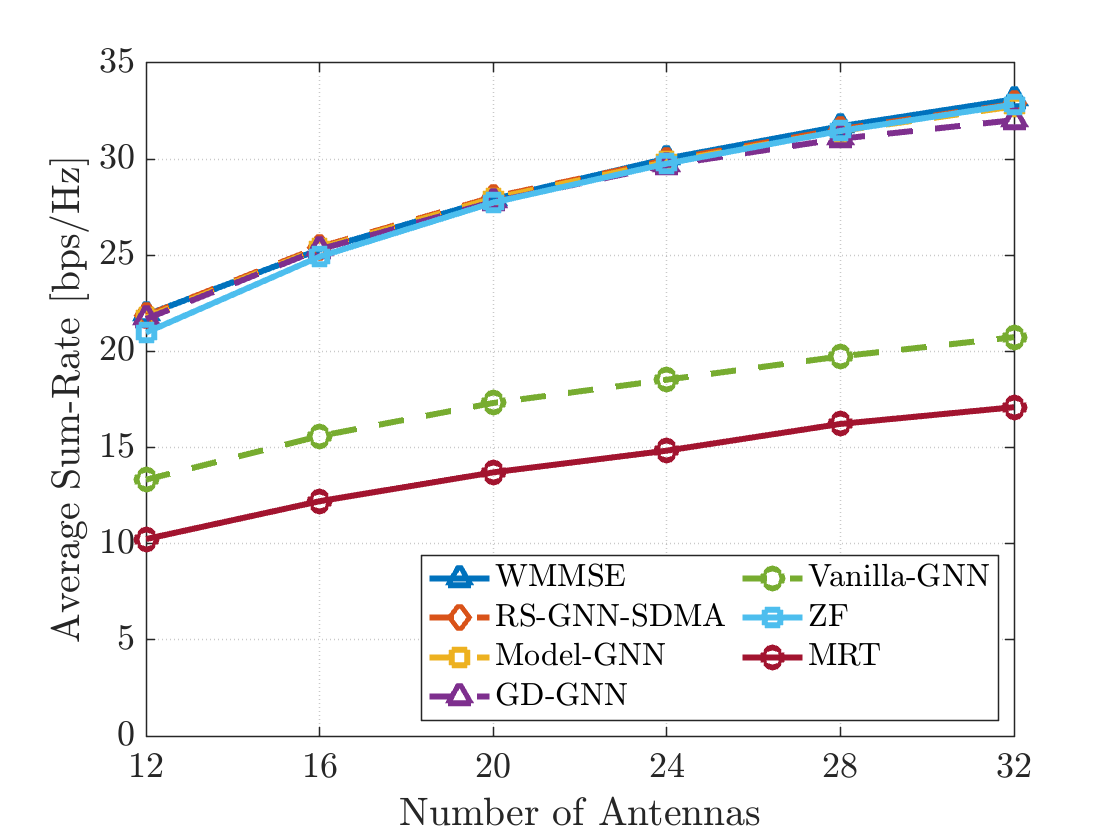}
        \caption{Antenna generalization}
        \label{Fig:7(a)}
    \end{subfigure}%
    \hfill
    \begin{subfigure}[t]{0.50\textwidth}
        \centering
        \includegraphics[width=\linewidth]{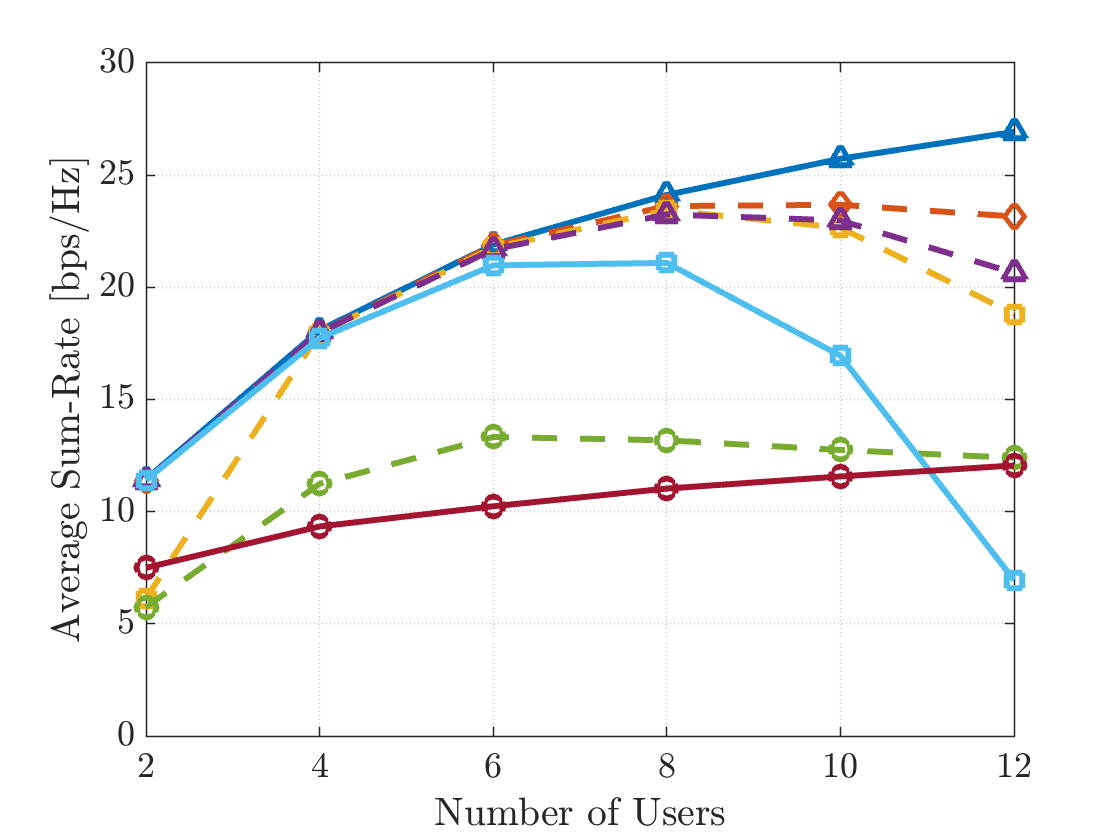}
        \caption{User generalization}
        \label{Fig:7(b)}
    \end{subfigure}%
    \hfill
    \begin{subfigure}[t]{0.50\textwidth}
        \centering
        \includegraphics[width=\linewidth]{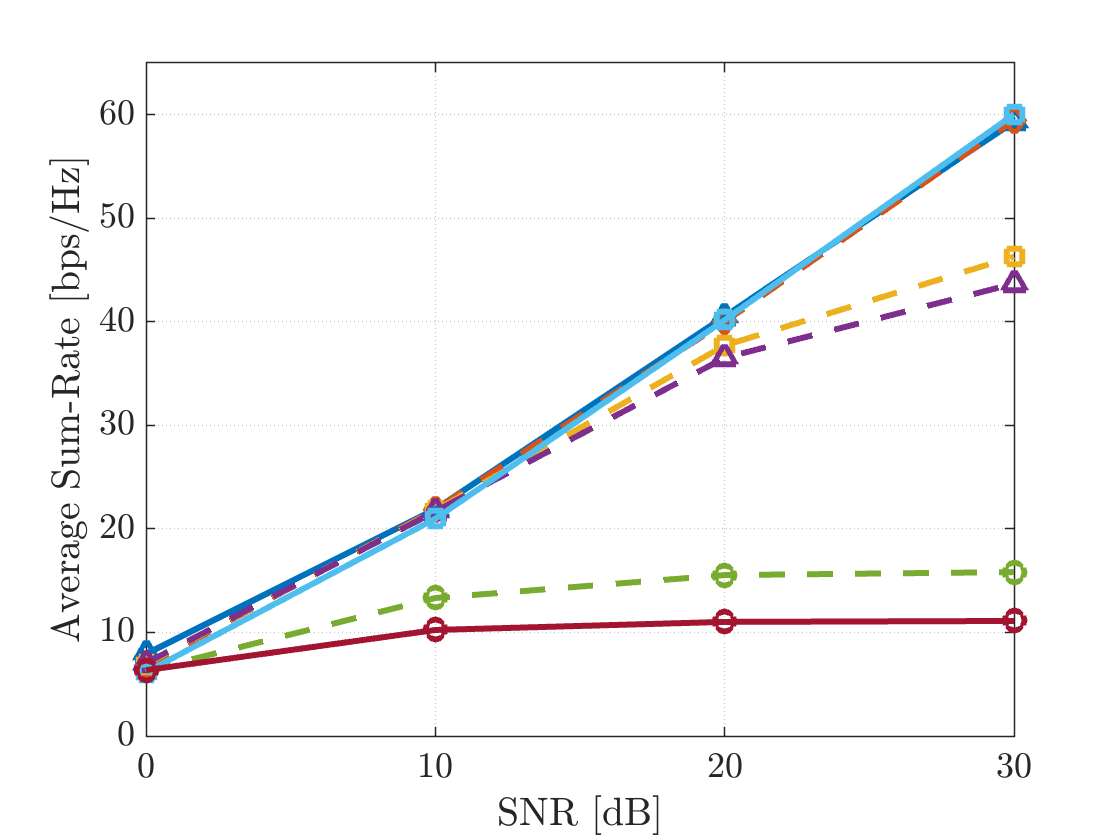}
        \caption{SNR generalization}
        \label{Fig:7(c)}
    \end{subfigure}%
    \hfill
    \begin{subfigure}[t]{0.50\textwidth}
        \centering
        \includegraphics[width=\linewidth]{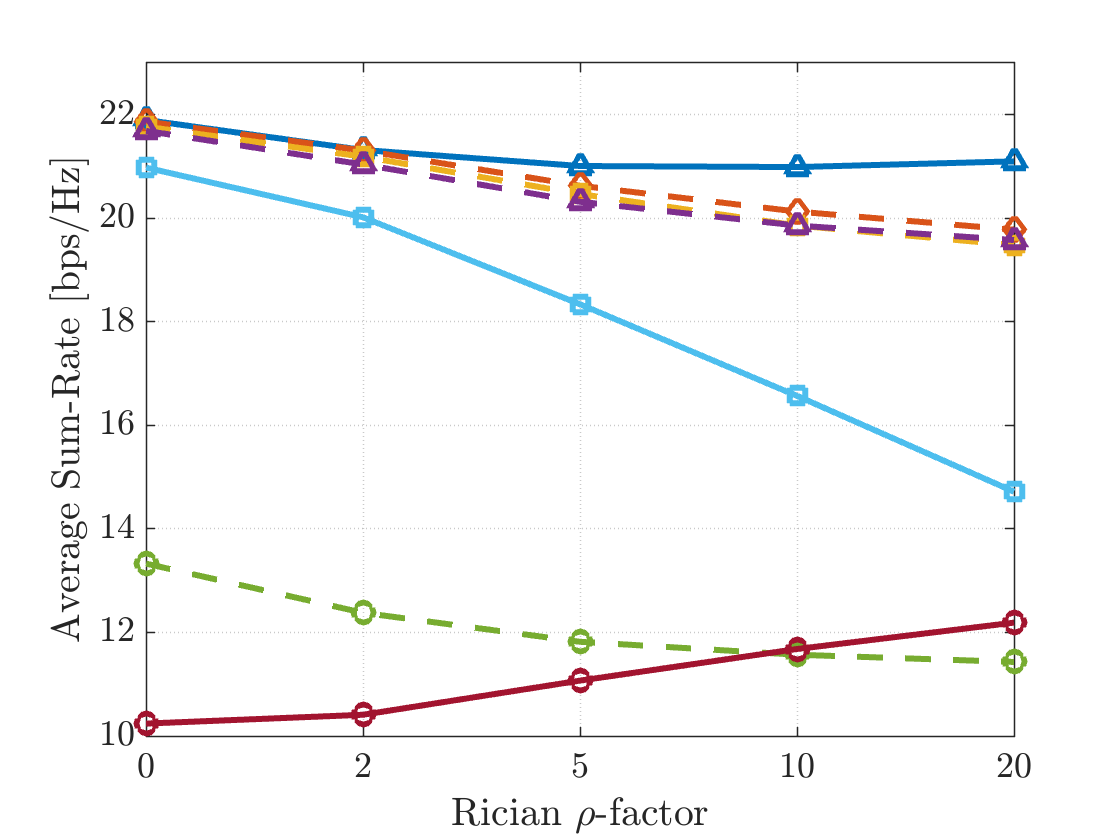}
        \caption{Channel-model mismatch}
        \label{Fig:7(d)}
    \end{subfigure}
    \caption{SDMA special case validation of the proposed RS-GNN against representative GNN-based precoding baselines. The subfigures show generalization over (a) antennas, (b) users, (c) SNRs, and (d) Rician $\rho$-factors.}
    \label{Fig:7}
\end{figure*}

\subsection{Ablation Study}

We conduct ablation experiments to validate the cumulative training strategy and to demonstrate the size generalizability of RS-GNN enabled by structure-preserving reconstruction. 

Fig.~\ref{Fig:3} examines the effect of cumulative recurrent training on variable-depth inference. RS-GNN is trained with $L_{\rm train}=5$ refinement steps and evaluated at inference depths $L_{\rm inf}\leq L_{\rm train}$, comparing two training strategies: {\em final-only training} (labeled RS-GNN-only in Fig.~\ref{Fig:3}), where the loss is applied only to the last recurrent output, and {\em cumulative training} (labeled RS-GNN-cumulative), where all intermediate outputs are explicitly supervised through the RSMA sum-rate objective in~\eqref{eq:loss}. RSMA-WMMSE is run until convergence and shown as a horizontal reference line, since its performance does not depend on the inference depth $L_{\rm inf}$. With final-only training, intermediate outputs are not directly optimized and exhibit unstable performance at inference depths shorter than $L_{\rm train}$. Cumulative training, by contrast, produces a stable and monotonically improving refinement trajectory as \(L_{\rm inf}\) increases, confirming that the same trained model can be evaluated at any \(L_{\rm inf}\leq L_{\rm train}\) without retraining. Consequently, \(L_{\rm inf}\) serves as a controllable knob for the online computational budget: a smaller \(L_{\rm inf}\) reduces inference time, whereas a larger \(L_{\rm inf}\) yields a more refined precoding solution.

Fig.~\ref{fig:4} evaluates the size-generalization capability of RS-GNN. The model is trained only at the base configuration $(N,K)=(8,6)$ and directly evaluated on unseen antenna and user configurations without retraining or architectural modification, whereas conventional learning-based precoders typically require retraining when the system dimensions change. RS-BNN-U is excluded from this size-generalization evaluation because its dense unfolded network is tied to the training system dimensions and would require separate retraining for each $(N,K)$ configuration. In contrast, RSMA-WMMSE and FP-HFPI are iterative numerical algorithms that are inherently size-agnostic, requiring no training at all, and therefore serve as consistent performance references across all tested configurations. RS-GNN closely tracks both references across the tested range of system sizes. In the antenna-size sweep, the proposed method maintains a sum-rate close to that of RSMA-WMMSE as $N$ increases, confirming that the architecture is not tied to the training antenna dimension. In the user-size sweep, RS-GNN remains competitive performance for both smaller and larger numbers of users. As $K$ increases with fixed $N$, the common-rate bottleneck and multi-user interference coupling become more severe, since more users must share the same spatial resources. Nevertheless, RS-GNN still achieves a substantial fraction of the RSMA-WMMSE performance without retraining. The joint \((N,K)\) sweep further confirms this behavior when both dimensions vary simultaneously. Although this setting is more demanding, RS-GNN preserves the overall scaling behavior of RSMA-WMMSE across all tested configurations, demonstrating that the learned refinement rule generalizes beyond the training configuration under simultaneous antenna and user size shifts.

\subsection{Comparison with RSMA Baselines}
\label{subsec:RSMA}

Fig.~\ref{Fig:5} compares the performance--online-inference-time tradeoff of the considered RSMA precoding methods, where online inference time denotes the runtime for generating a precoder per channel realization. Since the per-step algebraic costs differ across methods, comparing them at the same iteration or layer index can be misleading. We therefore vary the computational budget of each method and plot the achieved average sum-rate against the measured online inference time; operating points closer to the upper-left corner indicate a more favorable tradeoff. RSMA-WMMSE achieves strong sum-rate performance but incurs the largest online inference time, as each iteration requires solving a convex precoding subproblem. FP-HFPI reduces this cost by exploiting the optimal RSMA precoding structure and replacing the high-dimensional precoder search with a low-dimensional structural-variable search, though it still requires an online fixed-point procedure for the structural variables at every channel realization. RS-BNN-U further reduces the online cost by replacing the iterative dual-variable update of FP-HFPI with a dense unfolded neural update, while retaining the same structure-based reconstruction. However, since RS-BNN-U relies on a fixed-dimensional dense neural update inherited from the FP-HFPI-inspired unfolded structure, its tradeoff remains less favorable than that of RS-GNN. In contrast, RS-GNN predicts the structural variables through precoder-dependent graph message passing, avoiding any per-channel fixed-point search while retaining the structure-based reconstruction. 
{At the target average sum-rate of $19.5$ bps/Hz, RS-GNN reduces the online inference time by approximately $77.7\%$, $76.3\%$, and $98.9\%$ compared with RS-BNN-U, FP-HFPI, and RSMA-WMMSE, respectively.} As a result, RS-GNN achieves the most favorable performance--online-inference-time tradeoff among the considered methods.

Fig.~\ref{Fig:6} shows the empirical cumulative distribution function~(CDF) of the achieved sum-rate over independent test channel realizations. RS-GNN closely tracks the distributions of RSMA-WMMSE and FP-HFPI across most percentiles, demonstrating that the proposed method maintains near-optimal performance across the entire channel distribution, not merely on average. Compared with RS-BNN-U, the RS-GNN CDF shifts toward the higher-rate region in the low-to-middle percentile range, indicating improved robustness against unfavorable channel realizations. This result confirms that the recurrent graph-based refinement and precoder-dependent graph feature reconstruction yield more reliable RSMA precoding than the dense unfolded baseline.

\subsection{SDMA Special Case: Comparison with GNN-Based Precoders}\label{subsec:SDMA}

As established in Section~\ref{subsec:PEA}, existing GNN-based SDMA precoders cannot be directly extended to RSMA, since they lack a mechanism for generating the user-invariant common precoder required by the mixed permutation structure of RSMA. RS-GNN, in contrast, is designed for RSMA and reduces to a valid SDMA precoder simply by deactivating the common-stream branch, without any architectural modification. This asymmetry motivates the present comparison: we evaluate the SDMA special case of RS-GNN, referred to as {\bf RS-GNN-SDMA}, against representative GNN-based SDMA precoders to examine whether the proposed structure-preserving refinement remains effective even under the conventional private-stream-only formulation for which these baselines were originally designed.

The compared methods are as follows. \textbf{Vanilla-GNN~\cite{shen2022graph}} learns precoding through standard message passing without domain-specific structural information. \textbf{Model-GNN~\cite{guo2023model}} augments GNN features with structural information derived from the wireless precoding model. \textbf{GD-GNN~\cite{zhang2025gradient}} designs its message passing mechanism based on the gradient structure of the precoding objective. \textbf{ZF} and \textbf{MRT} are included as conventional linear precoding baselines: ZF suppresses multi-user interference via channel pseudo-inversion, while MRT aligns each beam with the corresponding user channel. \textbf{WMMSE}~\cite{shi2011iteratively} serves as a near-optimal high-complexity reference benchmark for SDMA sum-rate maximization. For a fair comparison with the SDMA benchmark methods, Fig.~\ref{Fig:7} follows the Rayleigh fading setup commonly adopted in the corresponding GNN-based SDMA precoding studies, rather than the correlated Rician fading model.

Fig.~\ref{Fig:7} presents the validation results under four evaluation scenarios. In Fig.~\ref{Fig:7}(a), the number of antennas is varied while the model is evaluated without retraining. RS-GNN-SDMA closely tracks WMMSE performance across the entire antenna range, confirming that the structure-based reconstruction is not tied to the training antenna dimension. In Fig.~\ref{Fig:7}(b), the number of users is varied, and RS-GNN-SDMA maintains competitive performance for both smaller and larger user configurations. Although the performance gap from WMMSE widens under high user load, RS-GNN-SDMA consistently outperforms conventional GNN baselines, confirming that predicting low-dimensional structural variables is more robust than directly generating high-dimensional precoding vectors. Fig.~\ref{Fig:7}(c) evaluates SNR generalization. At high SNR, where interference management becomes increasingly critical, RS-GNN-SDMA preserves the scaling behavior of WMMSE, whereas several GNN baselines saturate or exhibit noticeable performance degradation. 
Fig.~\ref{Fig:7}(d) investigates robustness to channel distribution mismatch by varying the Rician $\rho$-factor. As $\rho$ increases and the channel statistics deviate from the Rayleigh training distribution, RS-GNN-SDMA remains close to WMMSE and consistently outperforms the GNN-based baselines, indicating that the proposed architecture does not overfit to the training channel distribution. Collectively, these results demonstrate that the precoder-dependent graph feature reconstruction and structure-based recurrent refinement are beneficial not only for RSMA precoding but also for its SDMA special case.

\section{Conclusion}
\label{sec:conclusion}
In this paper, we proposed RS-GNN, a recurrent structure-preserving graph neural network for size-scalable MU-MISO RSMA precoding. To the best of our knowledge, RS-GNN is the first GNN-based RSMA precoder to jointly generalize across arbitrary unseen numbers of antennas and users with a single trained model.  By embedding  the analytical beamforming structure into the recurrent refinement process, RS-GNN enables interference-aware precoder updates over a structured and compact solution space, with provable mixed permutation equivariance and size generalizability. Moreover, the same trained model supports variable-depth inference, allowing the performance--online-inference-time tradeoff to be directly controlled at deployment time without retraining. Unlike existing GNN-based SDMA precoders, which cannot be directly extended to RSMA due to the lack of a mechanism for generating the user-invariant common precoder, RS-GNN reduces to a conventional SDMA precoder simply by deactivating the common-stream branch, without any architectural modification.
Simulation results demonstrated near-WMMSE sum-rate performance with substantially lower inference time, robust generalization across varying system configurations, SNR conditions, and channel distributions, and consistently superior performance over existing GNN-based baselines under the SDMA special case.

\appendices
\section{Proof of Theorem 1}
\label{sec:apxA}
 We proceed by induction on the $L$ refinement layers.

\noindent\textit{\bf Base case} $(\ell=0)$: The initial precoder \(\mathbf W^{(0)}\) is given by~\eqref{eq:init}, where the private precoders follow matched-filter transmission and the common precoder is initialized along the aggregate channel direction  $\sum_{k\in\mathcal{K}}\mathbf{h}_k=\mathbf{H}\mathbf{1}_K$.
Since $\boldsymbol{\Pi}_{\rm U}\mathbf{1}_K=\mathbf{1}_K$ for any permutation matrix, we have
\begin{equation}
\bar{\mathbf{H}}\mathbf{1}_K
=\boldsymbol{\Pi}_{\rm A}\mathbf{H}\boldsymbol{\Pi}_{\rm U}\mathbf{1}_K
=\boldsymbol{\Pi}_{\rm A}\mathbf{H}\mathbf{1}_K,
\end{equation}
so that the aggregate channel direction is permuted by $\boldsymbol{\Pi}_{\rm A}$ only, independent of $\boldsymbol{\Pi}_{\rm U}$. Together with $\bar{\mathbf{H}}=\boldsymbol{\Pi}_{\rm A}\mathbf{H}\boldsymbol{\Pi}_{\rm U}$ for the private precoders, the permuted initializations therefore satisfy
\begin{align}
  \bar{\mathbf{w}}_{\rm c}^{(0)}
  &=\boldsymbol{\Pi}_{\rm A}\mathbf{w}_{\rm c}^{(0)},\nonumber\\
  \bar{\mathbf{W}}_{\rm p}^{(0)}
  &=\boldsymbol{\Pi}_{\rm A}\mathbf{W}_{\rm p}^{(0)}\boldsymbol{\Pi}_{\rm U},
\end{align} establishing the base case.

\noindent\textit{\bf Induction step}:
Assume that the permutation relation holds at a layer $\ell$, i.e.,
\begin{align}
  \bar{\mathbf{w}}_{\rm c}^{(\ell)}&=\boldsymbol{\Pi}_{\rm A}\mathbf{w}_{\rm c}^{(\ell)},\nonumber\\
  \bar{\mathbf{W}}_{\rm p}^{(\ell)}
  &=\boldsymbol{\Pi}_{\rm A}\mathbf{W}_{\rm p}^{(\ell)}\boldsymbol{\Pi}_{\rm U}.
  \label{eq:pe_induction_assumption}
\end{align}
We show that the same relation holds at the layer $\ell+1$ by tracing each step of the recurrent refinement.

\textit{Step 1. (Graph Feature Construction)}: Under~\eqref{eq:pe_induction_assumption}, the components in~\eqref{eq:rsma_edgefeature} transform as
\begin{align}
  \bar{\mathbf{H}}^{\rm H}\bar{\mathbf{H}}
  &=\boldsymbol{\Pi}_{\rm U}^{\rm T}
  \mathbf{H}^{\rm H}\mathbf{H}\boldsymbol{\Pi}_{\rm U},
  \label{eq:pe_gram}\\
  \bar{\mathbf{H}}^{\rm H}\bar{\mathbf{w}}_{\rm c}^{(\ell)}
  &=\boldsymbol{\Pi}_{\rm U}^{\rm T}
  \mathbf{H}^{\rm H}\mathbf{w}_{\rm c}^{(\ell)},
  \label{eq:pe_common}\\
  \bar{\mathbf{H}}^{\rm H}\bar{\mathbf{W}}_{\rm p}^{(\ell)}
  &=\boldsymbol{\Pi}_{\rm U}^{\rm T}
  \mathbf{H}^{\rm H}\mathbf{W}_{\rm p}^{(\ell)}\boldsymbol{\Pi}_{\rm U}.
  \label{eq:pe_private}
\end{align}
Hence, the graph feature tensor is re-indexed by $\boldsymbol{\Pi}_{\rm U}$ and is invariant to the antenna permutation~$\boldsymbol{\Pi}_{\rm A}$. The same relations hold for the self-loop feature $\mathbf{g}_{k\to k}^{(\ell)}$, obtained by setting $v=k$ in~\eqref{eq:rsma_edgefeature}, so that the self-loop feature is permuted consistently with the target node index.

\textit{Step 2. (Structural Variable Prediction)}: Since the processing networks are shared across all edges and the attention pooling is permutation-invariant over incoming messages, the predicted structural variables are permuted consistently with the user ordering:
\begin{align}
  \bar{\boldsymbol{\lambda}}_{\rm c}^{(\ell)} &=\boldsymbol{\Pi}_{\rm U}^{\rm T}\boldsymbol{\lambda}_{\rm c}^{(\ell)},\quad \bar{\boldsymbol{\lambda}}_{\rm p}^{(\ell)}
  =\boldsymbol{\Pi}_{\rm U}^{\rm T}\boldsymbol{\lambda}_{\rm p}^{(\ell)},\nonumber\\
  \bar{\mathbf{d}}_{\rm c}^{(\ell)}
  &=\boldsymbol{\Pi}_{\rm U}^{\rm T}\mathbf{d}_{\rm c}^{(\ell)},\quad
  \bar{\mathbf{d}}_{\rm p}^{(\ell)}=\boldsymbol{\Pi}_{\rm U}^{\rm T}\mathbf{d}_{\rm p}^{(\ell)}.
  \label{eq:pe_d}
\end{align}
Although $\mathbf{d}_{\rm c}^{(\ell)}$ is associated with the common-stream, it is a coefficient vector in the user-channel basis and is therefore permuted consistently with the columns of~$\mathbf{H}$; the reconstructed common precoder itself will be shown to remain invariant to user reordering. The diagonal structural matrices consequently satisfy
\begin{align}
\operatorname{diag}\Big(\bar{\boldsymbol{\lambda}}_{\rm c}^{(\ell)}\Big)
  &=\boldsymbol{\Pi}_{\rm U}^{\rm T}
  \operatorname{diag}\Big(\boldsymbol{\lambda}_{\rm c}^{(\ell)}\Big)\boldsymbol{\Pi}_{\rm U},\nonumber\\
\operatorname{diag}\Big(\bar{\boldsymbol{\lambda}}_{\rm p}^{(\ell)}\Big)
  &=\boldsymbol{\Pi}_{\rm U}^{\rm T}
  \operatorname{diag}\Big(\boldsymbol{\lambda}_{\rm p}^{(\ell)}\Big)\boldsymbol{\Pi}_{\rm U}.
  \label{eq:pe_lambda_diag}
\end{align}

\textit{Step 3. (Precoder reconstruction)}: Substituting~\eqref{eq:pe_gram} and \eqref{eq:pe_lambda_diag} into~\eqref{eq:C}, the coefficient matrices for the permuted channel satisfy
\begin{align}
  \bar{\mathbf{C}}_{\rm c}^{(\ell)}
  &=\boldsymbol{\Pi}_{\rm U}^{\rm T}
  \mathbf{C}_{\rm c}^{(\ell)}\boldsymbol{\Pi}_{\rm U},\nonumber\\
  \bar{\mathbf{C}}_{\rm p}^{(\ell)}
  &=\boldsymbol{\Pi}_{\rm U}^{\rm T}
  \mathbf{C}_{\rm p}^{(\ell)}\boldsymbol{\Pi}_{\rm U}.
  \label{eq:pe_Cbar}
\end{align}
Applying these relations together with~\eqref{eq:pe_d} to the linear systems~\eqref{eq:linearsolve}, the solutions for the permuted input satisfy
\begin{align}
  \bar{\mathbf{x}}_{\rm c}^{(\ell)}
  &=\boldsymbol{\Pi}_{\rm U}^{\rm T}\mathbf{x}_{\rm c}^{(\ell)},\nonumber\\
  \bar{\mathbf{X}}_{\rm p}^{(\ell)}
  &=\boldsymbol{\Pi}_{\rm U}^{\rm T}
  \mathbf{X}_{\rm p}^{(\ell)}\boldsymbol{\Pi}_{\rm U}.
  \label{eq:pe_solution}
\end{align} Substituting into the reconstruction step~\eqref{eq:reconstruction} and using~$\bar{\mathbf{H}}=\boldsymbol{\Pi}_{\rm A} \mathbf{H}\boldsymbol{\Pi}_{\rm U}$, the unnormalized precoders transform as
\begin{align}
  \tilde{\bar{\mathbf{w}}}_{\rm c}^{(\ell+1)}
  &=\boldsymbol{\Pi}_{\rm A}\tilde{\mathbf{w}}_{\rm c}^{(\ell+1)},\nonumber\\
  \tilde{\bar{\mathbf{W}}}_{\rm p}^{(\ell+1)}
  &=\boldsymbol{\Pi}_{\rm A}
  \tilde{\mathbf{W}}_{\rm p}^{(\ell+1)}\boldsymbol{\Pi}_{\rm U}.
  \label{eq:pe_recon_compact}
\end{align}

\textit{Step 4. (Power normalization)}: Since it holds that $\big\|\tilde{\mathbf W}\big\|_{F}=\big\| \big[
\boldsymbol{\Pi}_{\rm A}\tilde{\mathbf w}_{\rm c},
\;
\boldsymbol{\Pi}_{\rm A}\tilde{\mathbf W}_{\rm p}
\boldsymbol{\Pi}_{\rm U}
\big]\big\|_{F}$ for any permutation matrices, the normalization factor in~\eqref{eq:powernormalization} is identical for $\tilde{\bar{\mathbf{W}}}^{(\ell+1)}$. The normalized precoders therefore satisfy
\begin{align}
  \bar{\mathbf{w}}_{\rm c}^{(\ell+1)}
  &=\boldsymbol{\Pi}_{\rm A}\mathbf{w}_{\rm c}^{(\ell+1)},\nonumber\\
  \bar{\mathbf{W}}_{\rm p}^{(\ell+1)}
  &=\boldsymbol{\Pi}_{\rm A}
  \mathbf{W}_{\rm p}^{(\ell+1)}\boldsymbol{\Pi}_{\rm U}.
\end{align} By induction, the proof is completed.



\end{document}

%% file: macros.tex
\long\def\comment#1{}

\newfont{\bbb}{msbm10 scaled 700}

\newfont{\bb}{msbm10 scaled 1100}
\newcommand{\CC}{\mbox{\bb C}}

\newcommand{\EE}{\mbox{\bb E}}

\newcommand{\hv}{{\bf h}}

\newcommand{\sv}{{\bf s}}

\newcommand{\wv}{{\bf w}}

\newcommand{\xv}{{\bf x}}

\newcommand{\Hm}{{\bf H}}
\newcommand{\Id}{{\bf I}}

\newcommand{\Wm}{{\bf W}}

\newcommand{\Ec}{{\cal E}}

\newcommand{\Kc}{{\cal K}}

\newcommand{\Oc}{{\cal O}}
\newcommand{\Pc}{{\cal P}}

\newcommand{\lambdav}{\hbox{\boldmath$\lambda$}}

\newcommand{\thetav}{\hbox{\boldmath$\theta$}}

\renewcommand{\Re}{{\rm Re}}
\renewcommand{\Im}{{\rm Im}}
